\documentclass{amsart}

\usepackage{lmodern}
\usepackage[T1]{fontenc}
\usepackage[english]{babel}

\usepackage[a4paper,top=4cm,bottom=2cm,left=4.25cm,right=4.25cm,marginparwidth=1.75cm]{geometry}

\usepackage{algorithm}
\usepackage{algpseudocode}

\usepackage{amsmath}
\usepackage{amssymb}
\usepackage{amsthm}
\usepackage{mathabx}
\usepackage{graphicx}
\usepackage{subcaption}

\usepackage{pgf}
\usepackage{tikz-cd}
\usepackage[foot]{amsaddr}

\usepackage{float}

\usepackage[section]{placeins}

\usepackage[normalem]{ulem} 

\usepackage[backend=biber,style=numeric-comp,doi=true,url=false,language=english]{biblatex}

\DeclareSourcemap{
  \maps[datatype=bibtex]{
    \map[overwrite]{
      \step[fieldsource=doi, final]
      \step[fieldset=url, null]
      \step[fieldset=eprint, null]
    }
    \map[overwrite=true]{
      \step[fieldset=language, null]
      \step[fieldset=issn, null]
      \step[fieldset=isbn, null]
    }
  }
}

\DeclareMathOperator{\NV}{NV}
\DeclareMathOperator{\SCC}{SCC}

\usepackage{csquotes}

\usepackage[colorlinks=true,allcolors=blue]{hyperref}
\usepackage{url}
\usepackage{doi}

\usepackage{orcidlink}

\newtheorem{theorem}{Theorem}[section]

\theoremstyle{definition}

\author[W.\ Jaworek]{Wojciech Jaworek$^{1,\ast,\orcidlink{0009-0002-8118-451X}}$}
\address{$^1$ Faculty of Applied Physics and  Mathematics,
Gda\'{n}sk University of Technology,
ul.\ Narutowicza 11/12, 80-233 Gda\'{n}sk,
Poland
}
\address{$^\ast$ Corresponding author}

\author[P.\ Pilarczyk]{Pawe\l{} Pilarczyk$^{1,2,\orcidlink{0000-0003-0597-697X}}$}
\address{$^2$ Digital Technologies Center,
Gda\'{n}sk University of Technology,
ul.\ Narutowicza 11/12, 80-233 Gda\'{n}sk,
Poland
}

\title[Chaotic Itinerancy]{Set-Oriented Approach to the Analysis of Chaotic Itinerancy}



\begin{document}

\begin{abstract}
Chaotic itinerancy (CI), brought to attention, among others, by K.~Ikeda, I.~Tsuda and K.~Kaneko in the early 1990s, is a phenomenon in which trajectories in a dynamical system experience periods of ordered motion near quasi-attractors interspersed with chaotic transitions between them.
Possible maps in which CI was found include coupled map lattices (CML) and globally coupled one-dimensional chaotic maps (GCM). We study such maps using numerical methods and graph algorithms.
Specifically, we partition the state space into a finite grid of compact subsets, and we represent the map using a multivalued mapping of grid elements. This mapping can be perceived as a directed graph, with grid elements as vertices and individual mappings between them as weighted edges. This setup provides a coarse view of global dynamics and opens the opportunity for using Markov chains and efficient graph algorithms to study dynamical features. In particular, invariant sets can be found by computing strongly connected components in the graph. Analysis of the transition matrix of the graph makes it possible to find its stationary distribution and to compute local entropy as a measure of expansion or instability in the system. Using these tools, we propose an algorithm for assessing whether a certain map possesses the CI property and show its application to dynamical systems: a globally coupled system of logistic maps and a variant of a CML system for which we conduct computations for a large range of parameters.
\end{abstract}

\maketitle

\noindent
\textbf{Chaotic itinerancy (CI) is a dynamical phenomenon on the verge between chaos and order. It may emerge in a dynamical system when some attractors lose their stability (and become so-called attractor ruins) but still keep enough strength to make trajectories stay in their vicinity for prolonged periods of time. Such trajectories wander chaotically in the phase space between visiting attractor ruins, and the observed dynamics is characterized by alternating periods of ordered and chaotic behavior. Despite its applications in neuroscience and artificial intelligence, the vast majority of methods for the investigation of this phenomenon focus on visual assessment of the results of numerical simulations of individual trajectories. We make an attempt to fill this gap and propose a set-oriented algorithmic numerical method for the analysis of global dynamics across the entire phase space at finite resolution, aimed at the detection and analysis of the CI phenomenon. The method is based on efficient graph algorithms and utilizes certain notions from the theory of Markov chains. The method allows one to find the location of attractor ruins in the phase space and provides quantitative assessment of selected features of CI.}

\section{Introduction}
\label{sec:intro}

Chaotic itinerancy (CI) is a form of high-dimensional deterministic dynamics in which an orbit itinerates among quasi-stationary, low-dimensional ordered states, called \emph{quasi-attractors} or \emph{attractor ruins}, through episodes of high-dimensional chaos. Each ordered phase corresponds to motion near a quasi-stable subset of the state space, where most directions are attracting. Such a set resembles an attractor, so it is called a quasi-attractor, and could have emerged through a bifurcation in which a true attractor (an asymptotically stable set) has become unstable, which justifies calling it an attractor ruin, especially that just after the bifurcation the instability is still weak. Nevertheless, because at least one direction is unstable, the trajectory eventually escapes, wanders chaotically throughout other regions of the phase space, and is then re-captured by the stable manifold of the same or another attractor ruin.  The process repeats, giving epochs of locally stable behavior interrupted by chaotic transitions. This process is reflected by large fluctuations and slow convergence of Lyapunov exponents \cite{Tsuda2003-dx} and by the prevalence of Milnor (non-asymptotically stable) attractors that nevertheless exert global attraction \cite{Milnor1985-xp}.

\subsection{Importance of the chaotic itinerancy phenomenon}
\label{sec:importance}

The concept of chaotic itinerancy was brought to attention, among others, by Ikeda \cite{Ikeda1989-zg}, Tsuda \cite{Tsuda1991-pf} and Kaneko \cite{Kaneko1990-fj}; see also \cite{Tsuda2013} for a comprehensive overview. The topic finds its most important applications in the field of neuroscience, to get deeper understanding of cognitive flexibility, associative memory, and perception, where the brain spontaneously transitions between quasi-stable states of activity \cite{Freeman1987-rm, Freeman2003-ch}. It is also a powerful tool in artificial intelligence and robotics, where it can be used to model spontaneous behavior \cite{Inoue2020-hg}.

\subsection{Mathematical understanding of chaotic itinerancy}
\label{sec:understanding}

Although the concept is relatively easy to understand on an intuitive level, several challenges arise from the mathematical point of view. Unlike conventional attractors, the quasi-stable states are not truly invariant. Their nature resists the description by standard notions of dynamical systems theory. The central problem is therefore to identify the structure of meaningful dynamical regions, such as attractor ruins and chaotic transitions connecting them, and to characterize the statistical properties of the resulting itinerant dynamics.

\subsection{State of the art}
\label{sec:state}

The phenomenon of CI has been investigated since the 1990's, yet---in contrast to the ``classic'' deterministic chaos \cite{Hirsch2012,Li_Yorke}---no strict definition for CI has been proposed so far, and no rigorous criteria for a system to possess this behavior have been stated, even though some characteristic features that can be associated to CI are listed in some papers, e.g., in \cite[Sec.~2]{Tsuda2004-mk}. The vast majority of previous research in which this phenomenon was investigated relied on numerical simulations and visual inspection of the resulting data \cite{Tsuda2003-dx, Kaneko1990-fj}. In \cite{Tsuda2003-dx}, the authors proposed the rate of convergence and fluctuations of Lyapunov exponents as a possible indication of the presence of chaotic itinerancy. In \cite{Sauer2003-xk}, a general rule for constructing such systems was introduced. In the most recent paper \cite{Mierski2025}, an automated numerical algorithm was proposed for the analysis of the CI phenomenon exhibited by a single trajectory. That method is based on the computation of entropy to notice periods of ordered and chaotic behavior, a clustering algorithm to find possible attractor ruins, and statistical tests to provide evidence of randomness of the transitions.

\subsection{Our contribution}
\label{sec:contrib}

We introduce a set-oriented method for the analysis of discrete-time dynamical systems that is aimed at the detection of the CI property. Our method is inspired by the approach to approximating the dynamics in which a finite grid of compact subsets is introduced in the state space and the dynamical system is analyzed as a mapping between these grid elements represented as a directed graph; see, e.g. \cite{Arai2009,miyaji-2016,Pilarczyk-2012,Pilarczyk2023,Falecki2026}. We extend this method to include edge weights in the directed graph, which opens new opportunities for a more in-depth analysis of the dynamics. We introduce a method for locating the attractor ruins by means of a filtration of the weighted graph. We use the PageRank algorithm \cite{Page1999} to compute the stationary distribution of the underlying Markov chain, and therefore we are able to compute the statistical properties of subsets of the finite grid. In particular, we use the local entropy of elements of the grid, the entropy rate of subsets of the grid corresponding to the detected attractor ruins and the outflow of them as the indices to quantify the CI within the system.

\subsection{Structure of the paper}
\label{sec:structure}

In Sec.~\ref{sec:CML}, we introduce the particular systems that we examine throughout the paper. Those include CML and GCM models. In Sec.~\ref{sec:sets}, we describe a process of building a finite resolution approximation of a dynamical system and a way to filter the resulting graph. Then in Sec.~\ref{sec:metrics}, we describe metrics that we compute for the analyzed system. These include eigenvalues of the transition matrix $P$ (Sec.~\ref{sec:matrices}), stationary distribution of $P$ (Sec.~\ref{sec:matrices}), local entropy and entropy rate (Sec.~\ref{sec:entropy}), strongly connected components (Sec.~\ref{sec:SCC}), outflow (Sec.~\ref{sec:SCC}), and variation of local entropy and of the density of the stationary distribution (Sec.~\ref{sec:var}). In Sec.~\ref{sec:models}, we show and discuss the results of applying the proposed framework to the analysis of selected dynamical systems described in Sec.~\ref{sec:CML}.

\section{Coupled map lattices and globally coupled maps}
\label{sec:CML}

A well-known class of systems exhibiting chaotic itinerancy are coupled map lattices (CMLs). CMLs provide a classic model for studying spatio-temporal complexity by modeling spatially extended systems as arrays of discrete-time one-dimensional dynamical systems, each with interactions with a selection of the others \cite{Kaneko1990-fj,Yanagita1995-jm}. Each individual one-dimensional system changes according to a local map, and coupling introduces interactions with neighboring maps, resulting in a combination of chaotic behavior of the individual maps and global, collective dynamics of the coupled system.

Formally, given continuous maps $f_i \colon X \to X$, with $X \subseteq \mathbb{R}$ and $i = 1, \ldots, n$, the CML of these maps is a semi-dynamical system on $X^n$ defined by
\begin{equation}\label{eq:CML}
    x_{t+1}^i = (1-\varepsilon) f_i(x_t^i) + \frac{\varepsilon}{|N(i)|} \sum_{j \in N(i)} f_j(x_t^j),
\end{equation}
where $x_t^i$ denotes the state of the $i$-th component at time $t$, the number $\varepsilon \in [0,1]$ is the coupling strength, and $N(i)$ denotes the set of indices of the components that influence the $i$-th component. The structure of $N(i)$ encodes the spatial complexity of the lattice, ranging from a few nearest neighbors of $i$ only to the entire set $\{1, \ldots, n\}$. Maps $f_i$ are typically taken as the logistic map $f(x)=1-ax^2$, and the type of dynamics is investigated as a function of the parameter $a$ and the coupling strength $\varepsilon$; see e.g.\ \cite{Kaneko1990-fj}.

The interplay between local chaotic dynamics and spatial coupling gives rise to a considerable range of different phenomena, including synchronization, pattern formation, spatiotemporal chaos, and itinerant motion between quasi-stable states; the last one is of our interest in this research. Of particular importance is the case where $N(i)$ contains all dimensions, leading to the globally coupled maps (GCM) model. This model was introduced and researched by Kaneko as a canonical model for high-dimensional dynamics in which he observed clustering phenomena and chaotic itinerancy \cite{Kaneko1991-fh,Kaneko1993-uy,Kaneko2003-oy,Kaneko2015-id}.

In the GCM model, chaotic itinerancy (CI) appears in terms of interleaving phases of synchronization and desynchronization between the one-dimensional components of the system. Subsets of the components join into clusters during a coherent phase, which is then broken by sudden reorganization caused by temporary desynchronization. We say that the CI appears in the system if the two phases are balanced over a long period of time.

Later studies showed that CI was not only present in the GCM model but was also observed in a broader class of CMLs. These systems include models of convection \cite{Yanagita1995-jm}, coupled logistic and Gaussian maps \cite{Tanaka2005-if,Kobayashi2018-yv}, and hybrid models linking physical bodies and neural oscillators \cite{Park2017-dl,Inoue2020-hg}.

In the sense explained above, CMLs and GCMs are fundamental models for understanding the CI phenomenon, where it emerges when balance is attained between local chaos and global coupling, allowing insight into the mechanisms of instability of attractor ruins and the structure of itinerant switching between them.

We shall focus on the following CML model~\cite{Tsuda2003-dx} built from one-dimensional maps of the circle defined as
\begin{equation}
\label{eq:tsudaMap}
f(x) = x-\omega \cos(2\pi x)+\omega \quad (\text{mod} \ 1)
\end{equation}
with a somewhat non-standard coupling given by
\begin{equation}
\label{eq:tsudaCML}
x^i_{t+1} = f(x^i_t) + \varepsilon \left(\sin(2\pi x^{i-1}_t) - 2\sin(2\pi x^{i}_t) + \sin(2\pi x^{i+1}_t) \right),
\end{equation}
where $i\in \mathbb{Z}_n$ is the index of a specific one-dimensional map, $\varepsilon$ is the coupling strength, and $\omega \in [0,1]$ is a nonlinearity parameter in $f$. This CML system was studied in \cite{Tsuda2003-dx} with $n=5$, $\omega=0.5$, and $\varepsilon = 0.08885$. In order to work with a lower-dimensional system, we fix $n=3$ and $\omega=0.5$ for this study, and we choose $\varepsilon=0.105958$, a coupling strength for which we empirically observed symptoms of CI in numerical simulations by examining trajectories in the phase space. Note that for $n=3$, in fact we have a system somewhat similar to a GCM system because each map is coupled with both others. With our choice of $\varepsilon$, equation (3) implies that $x_{t+1}^{i} \in [-0.5,1.5]$, because $f(x_t^i) \in [0,1]$ and $|\sin(\cdot)| \leq 1$; therefore, in what follows, we shall analyze the dynamics generated by this system on $[-0.5,1.5]^3$.

The advantage of choosing this specific 3-dimensional system instead of a typical GCM or CML model based on the logistic map is that one can already observe here the chaotic intinerancy phenomenon with a few attractor ruins in spite of the low dimension of the system. Such a low dimension helps considerably to apply our approach in which the phase space is subdivided into a uniform grid. This approach in higher dimensions is considerably more time- and memory-consuming.

\section{A coarse-grained approach to the analysis of dynamics}
\label{sec:sets}

In this section, we introduce a framework for analyzing dynamics of a given discrete-time semi-dynamical system (called a \emph{dynamical system} for short) induced by a continuous map $f\colon X\to X$. Instead of focusing on individual trajectories and their analytical properties, we partition the state space $X$ into a finite grid $\mathcal{X}$ of convex and compact subsets (such as $n$-dimensional cubes if $X \subset \mathbb{R}^n$) and define a multi-valued mapping $\mathcal{F}\colon \mathcal{X} \multimap \mathcal{X}$, perceived as a directed graph, to represent $f$. This approach is inspired by earlier work such as \cite{Arai2009,Luzzatto2011-us,Pilarczyk2023}. Now we extend it to include graphs with weighted edges.

Although the theoretical and algorithmic framework is designed to work with a general form of a grid, in order to avoid technical difficulties (especially in applications), we focus on cubical grids on a bounded subset of $\mathbb{R}^n$ and additionally assume that each grid element has the same volume. Note that by simple identification, this framework is also directly applicable to an $n$-torus perceived as the Cartesian product of $n$ circles.

Let the state space $X$ be an $n$-dimensional hypercube (or actually a cuboid) defined as the Cartesian product of intervals: $X = \prod_{i=1}^n [a_i, b_i]$. We divide each interval $[a_i, b_i]$ into $d_i$ subintervals of length $\delta_i = (b_i - a_i)/d_i$. The state space of the dynamical system that we analyze is thus discretized into a collection of $n$-dimensional cubes (or actually small cuboids, but we shall call them \emph{cubes} for short), where each cube is defined as $Q = \prod_{i=1}^n I_i$, where each $I_i$ is a subinterval of $[a_i, b_i]$ of length $\delta_i$ coming from the subdivision. This collection of cubes forms a discretized domain~$\mathcal{X}$. 

\subsection{Representing a map by a directed graph}
\label{sec:graph}

Given a grid $\mathcal{X}$ (either the cubical grid described above or another kind of grid), we construct a multi-valued map $\mathcal{F}$ to serve as a coarse-grained representation of the map $f$ that generates the dynamical system under consideration. For each $Q \in \mathcal{X}$, we uniformly sample a set $S(Q)$ of $N$ points from $Q$ and apply the map $f$ to each of the points to determine the set $\mathcal{F}(Q)$ as follows:
\begin{equation}
\mathcal{F} \colon \mathcal{X} \ni Q \mapsto \{R \in \mathcal{X} : f(S(Q)) \cap R \neq \emptyset\} \subset \mathcal{X}.
\end{equation}
Then we represent the multi-valued map $\mathcal{F}\colon \mathcal{X}\multimap \mathcal{X}$ by means of a directed graph $G=(V, E)$ as follows:
\begin{eqnarray}
V & = & \mathcal{X}, \nonumber \\
E & = & \{(Q,R) : Q,R \in \mathcal{X} \text{ and }  R \in \mathcal{F}(Q)\}. \nonumber
\end{eqnarray}
In this way, directed edges of $G$ reflect the mapping $\mathcal{F}$.

This framework proved to be effective in some previous applications. For example, strongly connected path components of a graph constructed similarly to our case were used in \cite{Arai2009} to find isolating neighborhoods, the shortest paths in such components were used in \cite{Pilarczyk2023} to quantify the recurrence time within such isolating neighborhoods, and global properties of the graph were used in \cite{Luzzatto2011-us} to prove that the system was transitive or mixing in a finite-resolution sense.

\subsection{Adding weights to graph edges}
\label{sec:weights}

We additionally introduce edge weights $w \colon E \to \mathbb{R}$ to record the volume of intersection of $f(Q)$ with $R$, where $E = (Q,R)$, effectively measured by the fraction of sample points mapping from $Q$ to $R$:
\begin{equation}\label{eq:weights}
w(e) = \text{card} (f(S(Q)) \cap R) / N, \text{ where } e = (Q,R) \in E.
\end{equation}
This idea is illustrated in Fig.~\ref{fig:mvmap_digraph}.

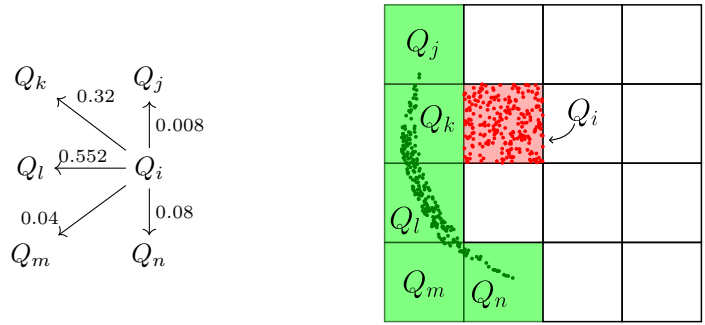
\begin{figure}[htbp]
    \centering
    \begin{minipage}{0.49\textwidth}
        \raggedleft
        \[\begin{tikzcd}[ampersand replacement=\&,cramped]
	{Q_k} \& {Q_j} \\
	{Q_l} \& {Q_i} \\
	{Q_m} \& {Q_n}
	\arrow["{{0.32}}"'{pos=0.8}, from=2-2, to=1-1]
	\arrow["{{0.008}}"', from=2-2, to=1-2]
	\arrow["{{0.552}}"'{pos=0.6}, from=2-2, to=2-1]
	\arrow["{{0.04}}"'{pos=0.9}, from=2-2, to=3-1]
	\arrow["{{0.08}}", from=2-2, to=3-2]
\end{tikzcd}
\]
    \end{minipage}
    \hfill
    \begin{minipage}{0.49\textwidth}
        \includegraphics[width=0.8\linewidth]{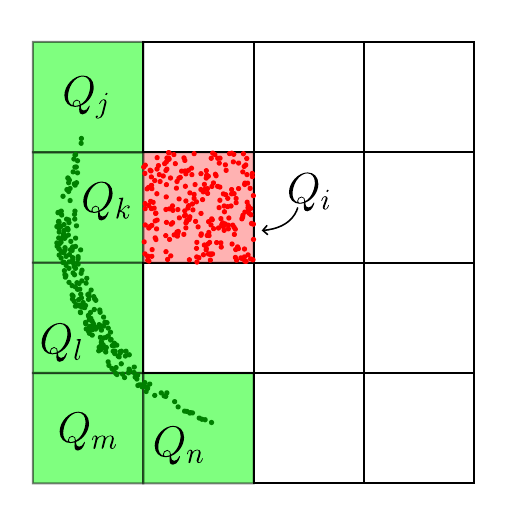}
    \end{minipage}
    \caption{Representing a multi-valued map on grid elements by means of a directed graph. The combinatorial image of the cube $Q_i$ (colored in red) consists of all cubes that contain the images of points sampled from $Q_i$ (colored in green). Edges indicate the mapping, and their weights were calculated using formula \eqref{eq:weights}.
    }
    \label{fig:mvmap_digraph}
\end{figure}

In what follows, we shall use the symbol $\mathcal{F}$ to refer to the multi-valued mapping and its associated graph as well.

The directed weighted graph $\mathcal{F}$ can be encoded in a row-stochastic transition matrix $P\in \mathbb{R}^{|V|\times |V|}$. For this purpose, assume $V = \{Q_1, \ldots, Q_{|V|}\}$, and define
\begin{equation}
\label{eq:P}
P_{ij} = 
\begin{cases}
w(e) & \text{if } e = (Q_i,Q_j) \in E, \\
0 & \text{otherwise}.
\end{cases}
\end{equation}

As a row-stochastic matrix, $P$ has some nice properties. In particular, it follows from the Gershgorin circle theorem that all the eigenvalues of $P$ have absolute values less than or equal to one. Note that $1$ is an eigenvalue of $P$ with the corresponding right eigenvector $\mathbf{1}$.

$P$ is usually a sparse matrix because each grid element is mapped to a limited number of other grid elements; therefore, from the computational perspective, it is sufficient to store information about its non-zero entries only.

We remark at this point that it is possible to compute the map $\mathcal{F}$ in a rigorous numerical way, for example, to ensure that $f(Q) \subset \bigcup \mathcal{F}(Q)$, and to compute rigorous lower and upper bounds on the measure of the intersection $f(Q) \cap R$ for each $R \in \mathcal{F}(Q)$ as the weights of the edges. For example, this could be achieved effectively using interval arithmetic and the procedures implemented in the CAPD software library; see \cite{Kapela2021} and \url{http://capd.ii.uj.edu.pl/}. However, this would introduce considerable technical complications and might be difficult to interpret since the difference between the lower and upper bounds might be considerable. Therefore, we choose to construct $\mathcal{F}$ by means of a numerical simulation in the spirit of Monte Carlo methods, as described above.

Representing the original dynamical system in the discretized state space $\mathcal{X}$ using the finite graph $\mathcal{F}$ provides a considerable advantage.
Namely, our framework translates questions about specific dynamical behavior of the system into problems of graph theory and linear algebra, where they can be solved with well-established and efficient algorithms and methods.
In particular, recurrent regions in the state space correspond to sets of grid elements that are connected by paths in both directions within the graph $\mathcal{F}$, called \emph{strongly connected path components} in~\cite{Arai2009}.

\subsection{Filtering the graph representation of a map}
\label{sec:filtering}

The main challenge in detecting CI using the framework we propose is the correct identification of core dynamical regions such as attractor ruins amid the complexity of the graph that represents the dynamical system. To focus on dynamically relevant transitions, we introduce the concept of filtration.

Let $\mathcal{F} = (V,E,w)$ be a weighted directed graph that represents a given dynamical system. Let $\tau\in[0,1]$. The \textit{$\tau$-filtration} of $\mathcal{F}$ is the graph $\mathcal{F_\tau} = (V_\tau, E_\tau, w_\tau)$ given by:
\begin{eqnarray}
\label{eq:filtE} E_\tau & = & \{e \in E : w(e) \geq \tau\} \\
\label{eq:filtV} V_\tau & = & \{v\in V : \exists\,u\in V \text{ such that } (u,v)\in E_\tau \text{ or } (v,u)\in E_\tau\} \\
\label{eq:filtw} w_\tau(e) & = & w(e) / \textstyle \sum \{w(e_x) : e_x=(u,x)\in E_\tau\}, \text{ where } e = (u,v) \in E_{\tau}
\end{eqnarray}

The intuition behind this definition is the following. Given a threshold $\tau \in [0,1]$, we create a subgraph of $\mathcal{F}$ in which we only keep the edges whose weights are at least $\tau$ and discard the vertices that are not endpoints of any such edge. For large values of $\tau$, the resulting filtered graph $\mathcal{F}_\tau$ highlights regions of the most important (heavy) transitions. For example, setting $\tau=1$ leaves only stable fixed points, that is, those grid elements that are entirely mapped to themselves.

Using filtration, we can simplify the structure of the graph to be analyzed. The described filtering does not remove a grid element $Q$ as long as it retains at least one incoming or outgoing edge of sufficient weight. Therefore, the choice of $\tau$ allows for the interpolation between detailed and coarse views of the system's dynamics.

We remark that Equations \eqref{eq:filtE} and \eqref{eq:filtV} cut out all non-relevant grid elements and connections (according to filtration threshold), while Equation \eqref{eq:filtw} rescales the weights to make sure that the transition matrix for the new graph is again row-stochastic.

As will be shown later, the filtration step allows us to extract collections of grid elements that correspond to attractor ruins. An important point is that we use the classic SCC algorithm for this purpose, which is very fast.

\section{Methods for the detection of the chaotic itinerancy phenomenon}
\label{sec:metrics}

While both chaotic itinerancy and deterministic chaos in the classic sense \cite{Hirsch2012, Li_Yorke} create complex and generally unpredictable trajectories, the dynamical behavior described as chaotic itinerancy is more subtle. The purpose of our work is to develop an algorithmic method to detect this behavior using a finite-resolution approximation of the dynamics. In particular, we would like to be able to differentiate chaotic itinerancy from ``classic'' chaos, in addition to distinguishing it from simple dynamics like the existence of a globally stable fixed point.

By representing the map $f$ as a weighted directed graph $\mathcal{F}$ that can be viewed as a finite-state Markov chain represented by means of a row-stochastic transition matrix $P$, we gain access to a broad range of mathematical tools. Among them are spectral, entropic, and structural quantities that allow us to distinguish between global chaos and chaotic itinerancy.

In the next subsections, we describe the following indicators and algorithms that we propose to use for the detection and quantification of the chaotic itinerancy phenomenon:
\begin{itemize}
    \item Stationary distribution of $P$: an indicator of where the trajectories spend most of the time. This distribution uncovers the location of quasi-attractors in the state space for visual inspection and is used by other indicators; see Sec.~\ref{sec:matrices}.
    \item Eigenvalues of the transition matrix $P$ of $\mathcal{F}$ and its spectral gap: indicators of stability and mixing speed for a Markov chain; see Sec.~\ref{sec:spectralgap}.
    \item Entropy and entropy rate: information-theoretic quantities that allow for the distinction of regions with ordered and chaotic dynamics in the state space; see Sec.~\ref{sec:entropy}.
    \item Strongly connected path components and outflow: graph algorithms to uncover invariant components, metastable sets (invariant sets for $\mathcal{F}_\tau$) and to measure transitions between them; see Sec.~\ref{sec:SCC}.
\end{itemize}
These quantities and methods should together provide enough information to classify the type of dynamics in a given model.

\subsection{Stationary distribution of the transition matrix}
\label{sec:matrices}

Since $P$ is a row-stochastic transition matrix for $\mathcal{F}$, we can compute its \emph{stationary distribution}, that is, a vector $\pi$ such that
\[
\pi P = \pi \quad \text{and} \quad \ \sum_i \pi_i = 1.
\]

An application of Brouwer's fixed point theorem implies that such a distribution exists. We show it here for the sake of completeness. Namely, let $\Delta \subset \mathbb{R}^n$ be the simplex of all vectors with non-negative entries that sum up to $1$. Consider the operator $\mathcal{P}\colon \Delta \ni \mu \mapsto \mu P \in \mathbb{R}^n$. Notice that $\mathcal{P}(\Delta) \subset\Delta$. Indeed, if $\sum_i \mu_i = 1$ then:
\[
    \sum_j (\mu P)_j = \sum_j\sum_i \mu_i P_{ij} = \sum_i \sum_j \mu_i P_{ij} = \sum_i \mu_i \sum_j P_{ij} = \sum_i \mu_i \cdot 1 = 1.
\]
The operator $\mathcal{P}$ is continuous and the simplex $\Delta$ is compact as a bounded and closed subset of $\mathbb{R}^n$, and it is obviously convex. Then Brouwer's fixed point theorem implies that there exists some $\pi\in\Delta$ such that $\pi = \pi P$.

However, it is not guaranteed that a stationary distribution is unique. In fact, a sufficient condition for the uniqueness of a stationary distribution is that the underlying Markov chain is irreducible and aperiodic \cite{Haggstrom_2002}.
Let us recall that a Markov chain is said to be \emph{irreducible} if every state is reachable from every other one, or, in terms of graph theory, all the vertices (states) constitute a single strongly connected component in the transition graph. We say that a state $V_i$ is \emph{periodic with period} $p$ if there exists a cycle that runs through the state $V_i$ and $p$ is the largest natural number that divides the length of every cycle that runs through $V_i$. A Markov chain is said to be \emph{aperiodic} if all its states are periodic with period $p=1$ \cite{Haggstrom_2002}.

The problem of finding a stationary distribution for a transition matrix reduces to applying an appropriate numerical algorithm, such as the well known power iteration method, to find an eigenvector of $P$ that corresponds to its largest eigenvalue $\lambda=1$. However, for the sake of effectiveness, we use a modification of the classic power iteration method inspired by the PageRank algorithm~\cite{Page1999}, introduced below as Algorithm~\ref{alg:power_iteration}. Note that the matrix $P$ is sparse, so the algorithm is fast.

\begin{algorithm}
\caption{PageRank-like Power Iteration}\label{alg:power_iteration}
\begin{algorithmic}
\State {\textbf{Input: } $P\in \mathbb{R}^{n\times n}$, $\alpha\in (0,1)$, \texttt{max\_iter} $\in \mathbb{N}$, \texttt{tol} $\in (0,1)$}
\State $\pi_0 \gets \text{vector with all entries } \frac{1}{n}$
\For{$k = 1,2, \ldots, \texttt{max\_iter}$}
\State $\hat{\pi}_k \gets (1-\alpha) \cdot (\pi_{k-1} P) + \alpha\cdot \frac{1}{n} \cdot \mathbf{1}$
\State $\pi_k \gets \frac{\hat{\pi}_k}{\lVert \hat{\pi}_k \rVert_1}$
\If{$\lVert\pi_k - \pi_{k-1}\rVert_1 < \texttt{tol}$}
\State $\text{break}$
\EndIf
\EndFor
\State \Return the last computed $\pi_k$
\end{algorithmic}
\end{algorithm}

\begin{theorem}
\label{thm:PageRank}
The sequence of vectors $(\pi_k)$ defined in Algorithm \ref{alg:power_iteration} converges to the unique left eigenvector of the matrix
\begin{equation}\label{eq:P_hat}
    \hat{P}=(1-\alpha) P + \frac{\alpha}{n} J,
\end{equation}
where $J$ is the square matrix of ones. Moreover, the matrix $\hat{P}$ is row-stochastic and all its entries are positive.
\end{theorem}

\begin{proof} 
Let us first notice that $\hat{P}$ is row-stochastic since it is a convex combination of two row-stochastic matrices. Indeed, for every $i$ we have
\[
\sum_{j=1}^n P_{ij} = 1 \quad \text{and } \quad \sum_{j=1}^n\frac{1}{n}J_{ij} = 1
\]
and thus 
\[
\forall_i \quad \sum_{j=1}^n \left((1-\alpha)P_{ij} + \frac{\alpha}{n}J_{ij}\right) = 1-\alpha+\alpha=1.
\]
Moreover, all the entries in $\hat{P}$ are strictly positive, as $\frac{1}{n}J$ itself has strictly positive entries and $P$ has non-negative entries. 

Let us now explain that in the classic power iteration method, starting with the initial vector $\pi_0 = [\frac{1}{n},\ldots,\frac{1}{n}]$, we would iteratively calculate $\pi_{k} = \frac{M\pi_{k-1}}{\lVert M\pi_{k-1}\rVert}$ to obtain a right eigenvector of a given matrix $M$. Finding a left eigenvector of $\hat{P}$, which we are interested in, is equivalent to finding the corresponding right eigenvector of the transposition of $\hat{P}$, which implies that computing $\pi_{k} = \frac{\pi_{k-1} \hat{P}}{\lVert \pi_{k-1} \hat{P}\rVert}$ instead would lead to the desired result.
However, in Algorithm~\ref{alg:power_iteration}, we iteratively calculate the vector
\begin{equation}
\label{eq:pagerank}
\hat{\pi}_{k} = (1-\alpha)\pi_{k-1} P + \alpha \frac{1}{n} \mathbf{1},
\end{equation}
where $\mathbf{1}$ is the vector of all 1's, and normalize it after each step to obtain $\pi_{k}$.
Let us show that this formula is equivalent to calculating $\hat{\pi}_k = \pi_{k-1} \hat{P}$. Indeed, 
\[
\pi \hat{P} = \pi \left( (1-\alpha)P + \frac{\alpha}{n}J\right) =(1-\alpha) \pi P + \frac{\alpha}{n} \pi J 
\]
and 
\[
\pi J = 
\begin{bmatrix}
\pi_1 & \pi_2 & \cdots & \pi_n
\end{bmatrix}
\begin{bmatrix}
1 & \cdots & 1 \\
\vdots & \ddots & \vdots \\
1 & \cdots & 1
\end{bmatrix}
= \begin{bmatrix}
    \sum \pi_i & \sum \pi_i & \ldots & \sum \pi_i
\end{bmatrix} = \mathbf{1}.
\]
With this iterative procedure, we obtain an estimate of a left eigenvector of $\hat{P}$. The matrix $\hat{P}$ is irreducible because all its entries are positive. For the same reason, the matrix $\hat{P}$ is also aperiodic, since there exists a directed edge from each state to itself with a positive weight, so the period of each state is exactly~$1$. This observation guarantees the uniqueness of a stationary distribution (left eigenvector) for $\hat{P}$.

Since all entries of the matrix $\hat{P}$ are strictly positive, the Perron-Frobenius theorem ensures that $\lambda_1$ (note that $\lambda_1 = 1$ because $\hat{P}$ is row-stochastic) is the unique largest eigenvalue of $\hat{P}$ and that the second largest in magnitude eigenvalue $\lambda_2$ of $\hat{P}$ satisfies $\lvert\lambda_2\rvert < \lvert\lambda_1\rvert = 1$. Moreover, there exists an eigenvector of $\hat{P}$ corresponding to $\lambda_1$ with all positive entries, which implies that the starting vector $\pi_0$ has a nontrivial component in the direction of this eigenvector. Recall that the convergence rate for the power iteration method is proportional to the quotient $\lvert\frac{\lambda_2}{\lambda_1}\rvert$. All this ensures the convergence in Algorithm~\ref{alg:power_iteration}.
\end{proof}

Since the vector $\pi$ obtained by Algorithm~\ref{alg:power_iteration} is calculated for a matrix that slightly differs from the true transition matrix $P$ for which we wanted to compute the stationary distribution, we assess the discrepancy in the result by calculating the $l_1$ norm of the difference between $\pi P$ and $\pi$. 

One of the reasons for using Algorithm~\ref{alg:power_iteration} instead of computing a stationary distribution for $P$ directly is that in the case of the existence of multiple attractors, a restriction to any of them would give rise to a valid stationary distribution that completely ignores the others. However, we are interested in capturing all the recurrent dynamics in a single stationary distribution. A simple schematic case of this type of dynamics is shown in Fig.~\ref{fig:reductible_graph}. For $\lambda=1$, we have three linearly independent vectors: \mbox{$\pi_1=[ 0,\ldots,0, 1]$}, \mbox{$\pi_2 = [
    \frac{1}{4}, \frac{1}{4}, \frac{1}{4}, \frac{1}{4}, 0, \ldots
]$} and \mbox{$\pi_3 =[0, \ldots, 0,\frac{1}{3}, \frac{1}{3},  \frac{1}{3}, 0]$}. With the use of Algorithm \ref{alg:power_iteration}, we get the vector $[\frac{1}{8}, \ldots, \frac{1}{8}]$ which includes all three cycles.
\begin{figure}
    \centering
\includegraphics[width=0.8\linewidth]{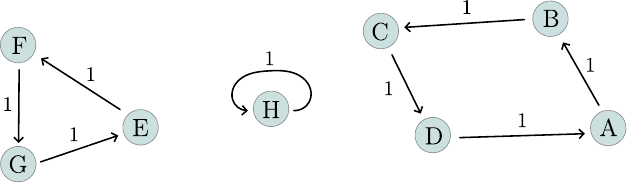}
    \caption{A reducible graph with $3$ linearly independent stationary vectors.}
    \label{fig:reductible_graph}
\end{figure}

In the finite-resolution setting, the vector $\pi$ provides a practical approximation of the time a trajectory would spend within each element of the grid. Since we expect that in the case of CI a trajectory would stay most of the time in close proximity of attractor ruins, we can directly visualize such regions by plotting $\pi$ on the grid using a color scale.

To illustrate the usefulness of the stationary distribution in understanding the dynamics, consider the CML system \eqref{eq:tsudaCML} with $n=3$, $\omega=0.5$ and $\varepsilon=0.105958$. We constructed a representation $\mathcal{F}$ of the map in the domain $A=[-0.5, 1.5]^3$ with each interval divided into $60$ subintervals of equal length, which resulted in $\mathcal{X}$ consisting of a total of $216\,000$ cubes. We used the sampling rate of $250$ points per cube to construct $\mathcal{F}$.

\begin{figure}[htbp]
    \centerline{
    \begin{subcaptionbox}{A sample trajectory consisting of $100\,000$ points plotted with opacity $0.3$.\label{fig:distr1}}[0.3\linewidth]
    {\includegraphics[height=\linewidth]{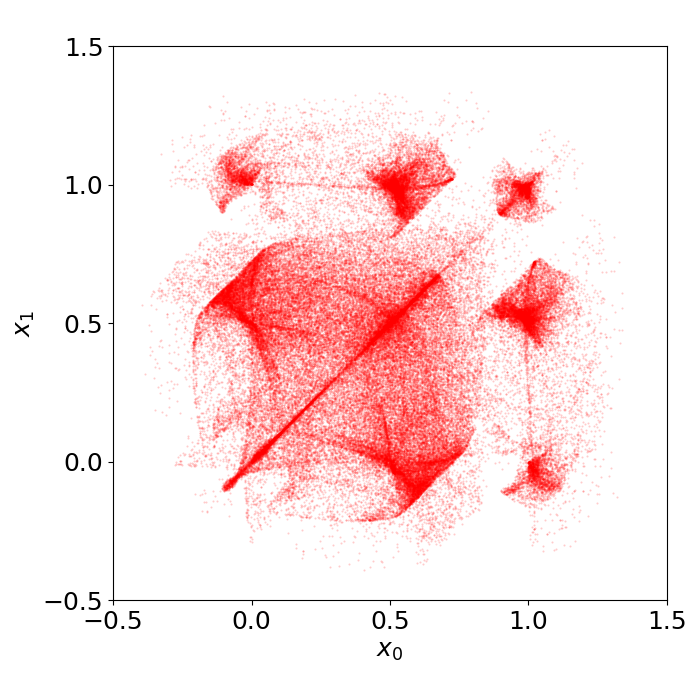}} 
    \end{subcaptionbox}
    \hspace{0.03\linewidth}
    \begin{subcaptionbox}{The starting points (green) and the endpoints (blue) of the $1\,000$ heaviest edges in $\mathcal{F}$.\label{fig:distr2}}[0.3\linewidth]
    {\includegraphics[height=\linewidth]{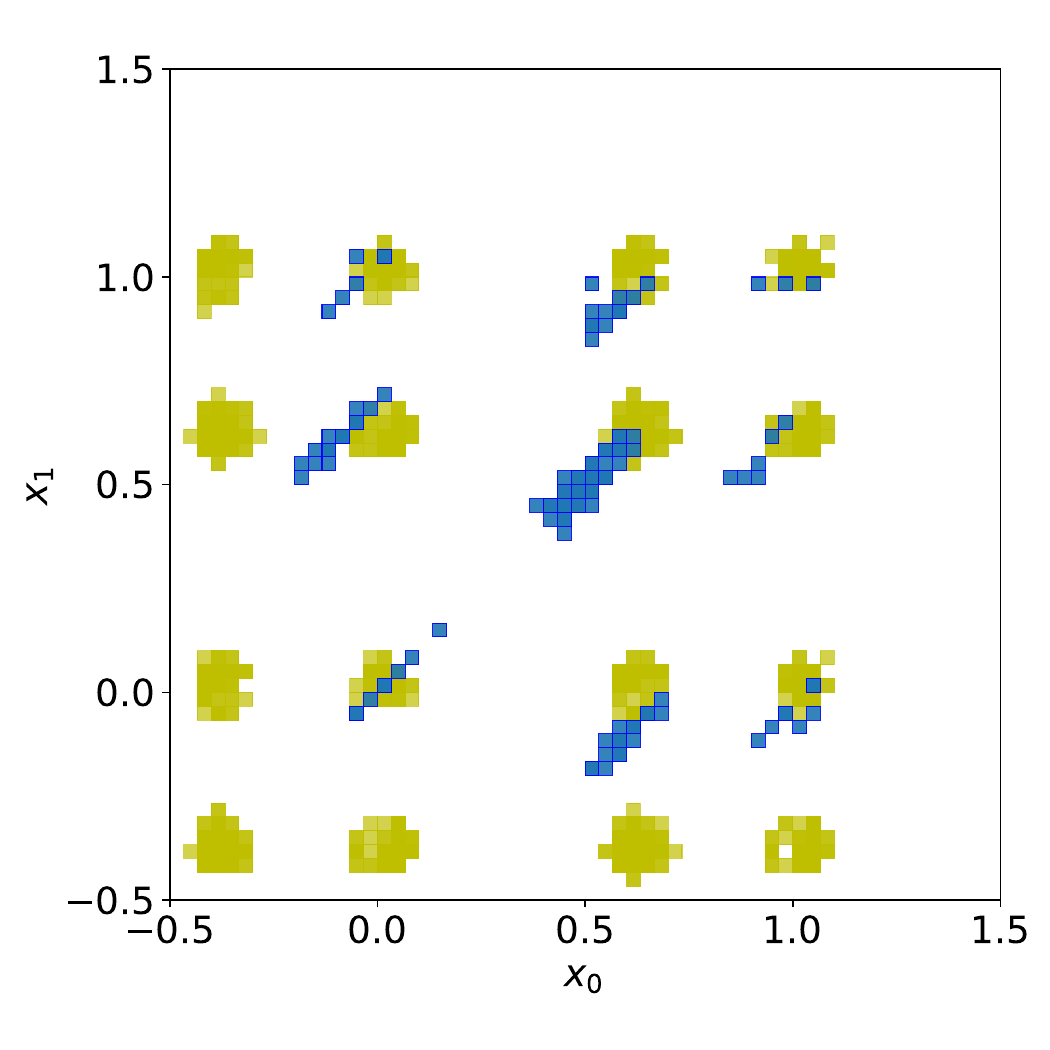}}
    \end{subcaptionbox}
    \hspace{0.03\linewidth}
    \begin{subcaptionbox}{The stationary distribution $\pi$. The darker the shade of gray, the higher the value.\label{fig:distr3}}[0.3\linewidth]
    {\includegraphics[height=\linewidth]{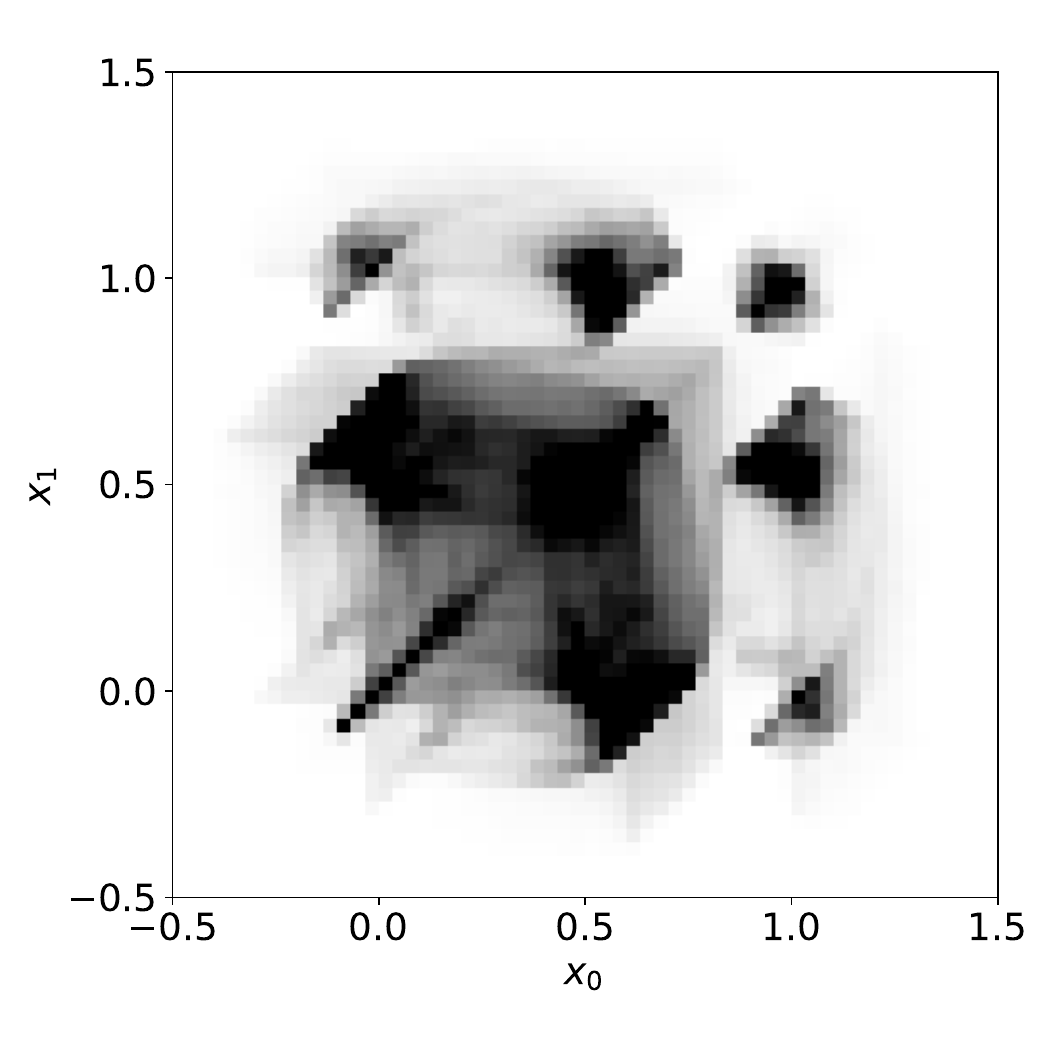}}
    \end{subcaptionbox}
    }
    \caption{Results of computations for a 3-dimensional CML system discussed in Sec.~\ref{sec:matrices}, projected onto the first two coordinates.}
    \label{fig:distribution}
\end{figure}

We first computed a sample trajectory for the system \eqref{eq:tsudaCML} consisting of $100\,000$ iterations of a single initial point to serve as a reference for further results. The points on the trajectory are shown in Fig.~\ref{fig:distr1}. One can clearly see some regions of high density. They apparently indicate the location of attractor ruins around which trajectories spend considerable amount of time. Between and around them, there are regions of lower but still positive density that might correspond to chaotic transitions.

In order to determine the most prominent connections in the graph $\mathcal{F}$, we found $1\,000$ heaviest edges in $\mathcal{F}$. In Fig.~\ref{fig:distr2}, we show cubes that constitute their starting points in one color and their endpoints in another color. The cubes corresponding to the endpoints are essentially in the same location as the clusters appearing in Fig.~\ref{fig:distr1}, and it is evident that the majority of transitions are located within these sets.

Finally, we computed an approximation of a stationary distribution $\pi$ of the transition matrix $P$ of $\mathcal{F}$ using Algorithm \ref{alg:power_iteration}. Its projection onto the first two coordinates is shown in Fig.~\ref{fig:distr3}. The values in the squares (shown using a color scale) are the sums of the distribution $\pi$ on all the cubes that share the first two coordinates. Potential attractor ruins have clearly higher probabilities than transient chaotic states.

\subsection{Spectral analysis}
\label{sec:spectralgap}

Another quantity in the analysis of the transition matrix $P$ is the spectral gap:
\begin{equation}
\label{eq:spectralGap}
\gamma = |\lambda_1| - |\lambda_2|,
\end{equation}
where $(\lambda_1=1, \lambda_2, \lambda_3, \ldots)$ are the eigenvalues of $P$ sorted in descending order by comparing their moduli.

The spectral gap can be used to approximate the mixing time of a Markov chain, which is the time that it takes until a chain gets close to its stationary state. A large gap indicates fast mixing and therefore the lack of metastable structures, while a small gap indicates the possible presence of metastable states or slow transitions between them \cite{Seabrook2023}. In the context of chaotic itinerancy, one would assume that the gap should be moderately low yet non-zero. By contrast, a finite representation of a dynamical system with no (quasi-)attractors typically exhibits a large spectral gap.

We explain that in the case of coexisting multiple attractors, the Markov chain is not irreducible. A stationary distribution that comprises all the attractors can be restricted to each attractor separately to give rise to a linearly independent set of eigenvectors corresponding to the eigenvalue~$1$. We then obtain $\lambda_2 = 1$ above and possibly more eigenvalues equal to $1$. The spectral gap is then equal to $0$.

For the calculation of the eigenvalues of $P$, we used the \texttt{eigs} function available in the \texttt{scipy.sparse.linalg} module of the SciPy Python package \cite{2020SciPy-NMeth}. This method is optimized for sparse matrices and works efficiently in our case.

\subsection{Entropy rate}
\label{sec:entropy}

Given a stochastic transition matrix $P$, the entropy rate of the corresponding Markov chain is a quantity that corresponds to the average unpredictability of transitions between the states of the Markov chain, which are the elements of the grid in our case; see \cite[Chapter~4]{Cover2006}. The entropy associated with transitions from a single grid element $Q_i$, which we also call \emph{local entropy}, is defined as
\begin{equation}
\label{eq:entropy1}
h_i = -\sum_{j\colon P_{ij}\not=0}P_{ij} \log P_{ij}.
\end{equation}
The \emph{entropy rate} of a time-homogeneous Markov chain with a finite number of states, such as $\mathcal{F}$ in our case, is expressed as the sum of entropies associated to elements of the grid averaged by the stationary distribution $\pi$:
\begin{equation}
\label{eq:entropy2}
h = \sum_i \pi_i h_i.
\end{equation}

It is a common practice to also introduce normalized entropy. We therefore define the \emph{normalized local entropy} and the \emph{normalized entropy rate} as follows:
\[
\hat{h}_i = \frac{h_i}{\log M} \quad\text{and}\quad \hat{h} = \sum_i \pi_i \hat{h}_i,
\]
where $M$ is the maximum number of states that could possibly be reached from the $i$-th state. For the Monte Carlo-like numerical simulation that we use, $M$ is the number of samples taken from each element of the grid $\mathcal{X}$ when computing $\mathcal{F}$, given that $|\mathcal{X}| \ge M$.

The purpose of normalization of the entropy rate is to obtain a number in the range $[0,1]$. The normalized entropy rate $\hat{h}$ is easier to interpret than plain $h$ because one can immediately tell whether the entropy rate is close to the theoretical maximum or stays in a low or high range.

Typically, a low entropy rate is characteristic to stable and fully deterministic dynamics, where transitions take place along easy to predict paths, such as attracting periodic orbits. On the other hand, a high entropy rate is typical for fully chaotic dynamics, where local expansion on the attractor is high, and therefore one grid element is typically mapped to several grid elements. In the case of chaotic itinerancy, one would expect that the entropy rate of the whole $\mathcal{F}$ will take intermediate values, because some regions have a low local entropy that reflects stability near attractor ruins, while other regions have a high local entropy that corresponds to unstable directions through which trajectories leave attractor ruins to chaotic transitions. Because of this, a high variation of local entropy (associated with each grid element) across the discretized state space should indicate the possible presence of chaotic itinerancy. On the other hand, the entropy rate restricted to a subset of grid elements that constitute a quasi-attractor should be significantly lower than the global one.

The local entropy computed for a coarse representation of the CML system \eqref{eq:tsudaCML} with $n=3$, $\omega=0.5$ and $\varepsilon=0.105958$ is shown in Fig.~\ref{fig:entropy}, projected onto the first two coordinates of the state space. The local entropy for all the cubes that shared the first two coordinates was added together. The color scale represents the percentiles of the total range of aggregated local entropy, where blue indicates the lower end of the range and red indicates the higher end of the range. We can distinguish transient, high-entropy chaotic states and low-entropy ordered sets here. When comparing Fig.~\ref{fig:entropy} with Fig.~\ref{fig:distribution}, it becomes apparent that the location of attractor ruins corresponds to the low-entropy subsets of the state space.

\begin{figure}[htbp]
    \centering
    \includegraphics[width=0.5\linewidth]{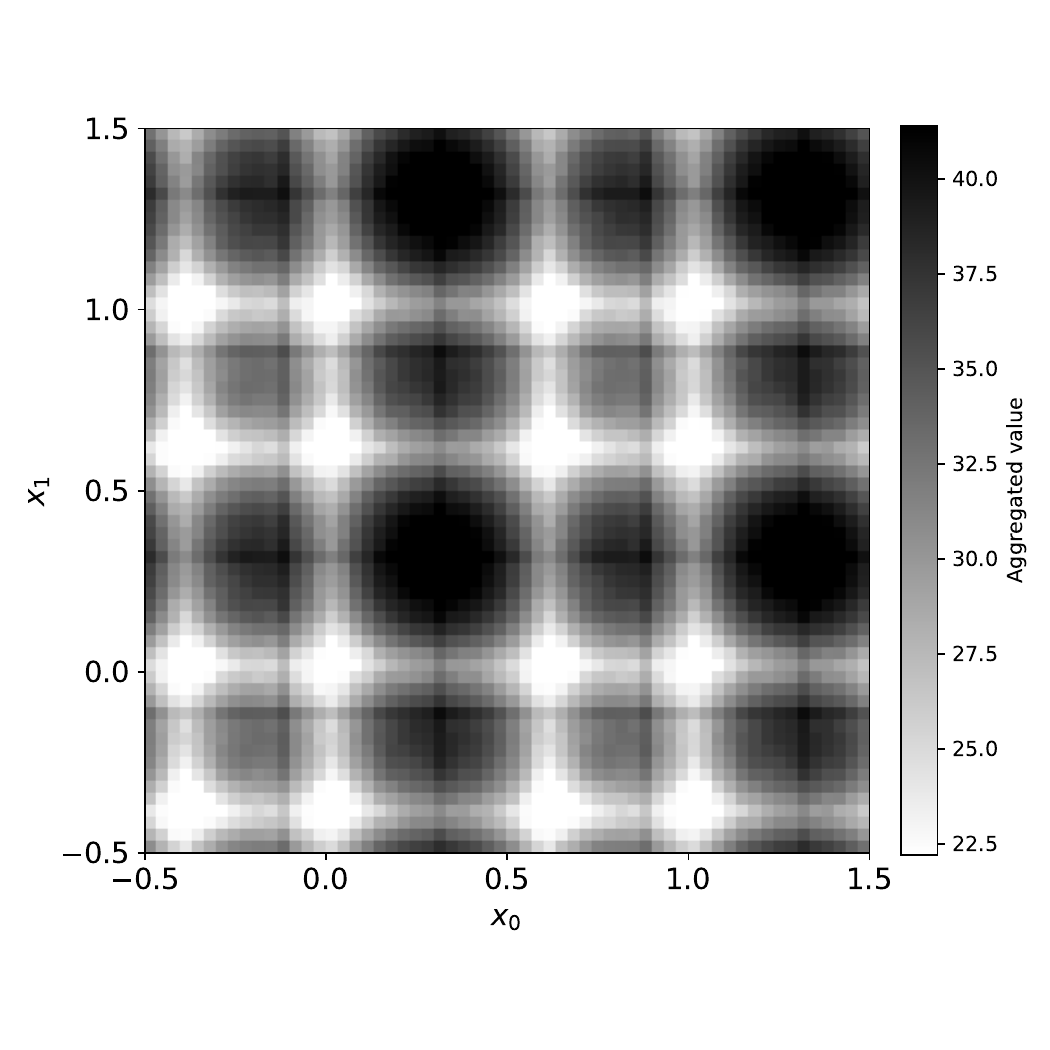}
    \caption{Aggregated local entropy computed for the 3-dimensional CML system discussed in Sec.~\ref{sec:entropy}.}
    \label{fig:entropy}
\end{figure}

\subsection{Strongly connected components and outflow}
\label{sec:SCC}

Let us now discuss the analysis of the directed graph $\mathcal{F}$ or its filtering $\mathcal{F}_\tau$ (see Sec.~\ref{sec:filtering}). We employ two notions through which we are going to analyze the dynamics: strongly connected components and outflow of the graph.

\emph{Strongly connected components} (SCCs) of a directed graph are maximal subsets of vertices with the property that every vertex is reachable from every other vertex along a directed path within the subset. In fact, we are interested in SCCs containing at least one edge, which are called \emph{strongly connected path components} in \cite{Arai2009}.

In the context of analysis of finite resolution dynamics, SCCs correspond to regions of the state space with recurrent or semi-recurrent dynamics. In particular, directly from the definition of the graph and an SCC, it follows that every two points in some elements of the grid corresponding to vertices in a common SCC, are connected by an $\varepsilon$-pseudo-orbit with $\varepsilon$ not exceeding the diameter of the grid elements. We recall that a pseudo-orbit is a sequence of points $x_1, \ldots, x_n$ such that $\mbox{dist}(f(x_i),x_{i+1}) \leq \varepsilon$, that is, it is an orbit in which we are allowed to make a correction of at most $\varepsilon$ at each iteration. For a connection between the two points, we may need to make an $\varepsilon$-correction before we make the first iteration, too. We remark that the notion of pseudo-orbits is investigated in particular in the context of shadowing: under certain assumptions, every pseudo-orbit is actually traced by a real orbit in the system; see e.g.~\cite{Koscielniak2014}.

When computing SCCs of a graph that represents $\mathcal{F}$ for a system that exhibits chaotic itinerancy, we should expect to obtain a single huge component containing all the attractor ruins and points in the region of chaotic transitions. However, when we consider a filtering $\mathcal{F}_\tau$ with an appropriate value of $\tau$, the intuition tells us that we should only obtain attractor ruins as SCCs, because connections in these regions are stronger than in the areas of chaotic transitions.

In order to analyze the stability of SCCs found in the graph of a filtering $\mathcal{F}_\tau$, we revert to $\mathcal{F}$ and use the concept of \emph{outflow}, motivated by the standard notion of graph conductance, to analyze the components found for $\mathcal{F}_\tau$. The concept of conductance was introduced in \cite{Jerrum1988} for time-reversible ergodic Markov chains.
Conductance is a metric that quantifies the diffusion of a given set.
A modification of the formula that was introduced in \cite{Kannan2004} to study the quality of spectral clustering is provided below.

The \emph{outflow} of a set $A\subset \mathcal{X}$ is the following quantity:
\begin{equation} \label{eq:conductance_our}
    \Phi(A) = \frac{\sum_{i\in A}\sum_{j\in \mathcal{X} \setminus A} \pi_i P_{ij}}{\pi(A)},
    \quad \text{where} \quad \pi(A) = \sum_{i\in A}\sum_{j\in \mathcal{X}} \pi_i P_{ij} = \sum_{i\in A} \pi_i.
\end{equation}

Intuitively, $\Phi(A)$ measures how easy it is to leave the set of grid elements~$A$. Low outflow means that very few trajectories escape $A$, which indicates possible stability of the set, whereas high outflow indicates that many trajectories escape from $A$, which indicates instability of $A$, possibly as part of chaotic dynamics present in CI.

Note that the denominator in \eqref{eq:conductance_our} was originally introduced in \cite{Kannan2004} as $\min\{\pi(A),\pi(\mathcal{X}\setminus A)\}$ to ensure that small sets do not have high conductance and, \textit{vice versa}, large sets do not have low conductance. However, our situation is different. We actually want to quantify the relative amount of $\pi$ that flows out of a set $A$ with respect to $\pi$ accumulated inside $A$, independent of the size of $A$.

To sum up, the computation of the SCCs of a filtering $\mathcal{F}_\tau$ of a finite-resolution representation $\mathcal{F}$ of the map $f$ of interest yields a collection of interesting regions where the dynamics concentrates, and the computation of their outflow provides information on their instability.

\subsection{Vitali variation}
\label{sec:var}

To quantify the variation of local entropy and of the density of the stationary distribution across the phase space, we refer to the notion of variation of a real-valued function and its generalization to functions of several variables introduced by Vitali; see e.g.~\cite{Leonov1996}. Since our space is discretized into a rectangular grid and the functions we consider are constant on the grid elements, we have a straightforward formula for the $n$-dimensional variation. We provide this formula below.

Recall that our phase space cuboid $X = \prod_{i=1}^{n} [a_i,b_i]$ is divided into the set $\mathcal{X}$ of $d_1 \times \cdots \times d_n$ smaller cuboids; see Sec.~\ref{sec:sets}. Consider a real-valued function $z \colon \mathcal{X} \to \mathbb{R}$. We shall represent each element of $\mathcal{X}$ by its index $(j_1, \ldots, j_n)$ with $j_i = 0, \ldots, d_i - 1$ as follows:

\begin{equation}
Q (j_1,\ldots,j_n) = \prod_{i=1}^{n} [a_i + \frac{j_i}{d_i}(b_i - a_i), a_i + \frac{j_i + 1}{d_i}(b_i - a_i)]
\end{equation}

Following \cite[equations (1) and (2)]{Leonov1996}, we first define the \emph{local $n$-dimensional variation} starting at $Q(j_1, \ldots, j_n)$, provided that $j_i < d_i - 1$ for all $i = 1, \ldots, n$, as follows:

\begin{equation}
\sigma_n (z, (j_1, \ldots, j_n)) = \sum_{\nu_1 = 0}^{1} \cdots \sum_{\nu_n = 0}^{1} (-1)^{\nu_1 + \cdots \nu_n} z(Q(j_1 + 1 - \nu_1, \ldots, j_n + 1 - \nu_n)),
\end{equation}
which reduces in the $2$-dimensional case to:
\begin{multline}
\sigma_2 (z, (j_1,j_2)) = z(Q(j_1 + 1, j_2 + 1)) - z(Q(j_1 + 1, j_2)) - \\ z(Q(j_1, j_2 + 1)) + z(Q(j_1, j_2)),
\end{multline}
and in the $3$-dimensional case to:
\begin{multline}
\sigma_3 (z, (j_1,j_2,j_3)) =
z(Q(j_1 + 1, j_2 + 1, j_3 + 1)) -
z(Q(j_1 + 1, j_2 + 1, j_3)) - \\
z(Q(j_1 + 1, j_2, j_3 + 1)) + 
z(Q(j_1 + 1, j_2, j_3)) -
z(Q(j_1, j_2 + 1, j_3 + 1)) + \\ 
z(Q(j_1, j_2 + 1, j_3)) +
z(Q(j_1, j_2, j_3 + 1)) -
z(Q(j_1, j_2, j_3)).
\end{multline}
Then we use the following formula for the \emph{$n$-dimensional variation} of the function $z$ on $\mathcal{X}$:
\begin{equation}
V_n (z, \mathcal{X}) = \sum_{j_1 = 0}^{d_1 - 2} \cdots \sum_{j_n = 0}^{d_n - 2} \left\lvert \sigma_n (z, (j_1, \ldots, j_n)) \right\rvert.
\end{equation}

Note that the value of $V(z,\mathcal{X)}$ depends not only on the function $z$, but also on the number of grid elements in $\mathcal{X}$. Therefore, to obtain a value that is easier to compare for different subdivisions of the same space and for functions with different ranges of values, we normalize this variation by dividing it by the total number $\lvert\mathcal{X}\rvert$ of grid elements and by the amplitude of $z$ defined as
\[
\text{amp}(z) = \max_{x \in \mathcal{X}} z(x) - \min_{x \in \mathcal{X}} z(x)
\]
and we define the \emph{normalized variation} of the function $z$ on $\mathcal{X}$ as follows:
\begin{equation}\label{eq:NormVariation}
\NV(z,\mathcal{X}) = \frac{V(z,\mathcal{X)}}{|\mathcal{X}|\cdot \text{amp}(z)}.
\end{equation}

As shown in Fig.~\ref{fig:V_and_NV}, normalized variation does not depend on the choice of the subdivision of the space, so we use this concept to measure the variation of stationary distribution density over $\mathcal{X}$, as well as the local entropy across the grid~$\mathcal{X}$.
\begin{figure}[htbp]
    \centering
    \includegraphics[width=0.5\linewidth]{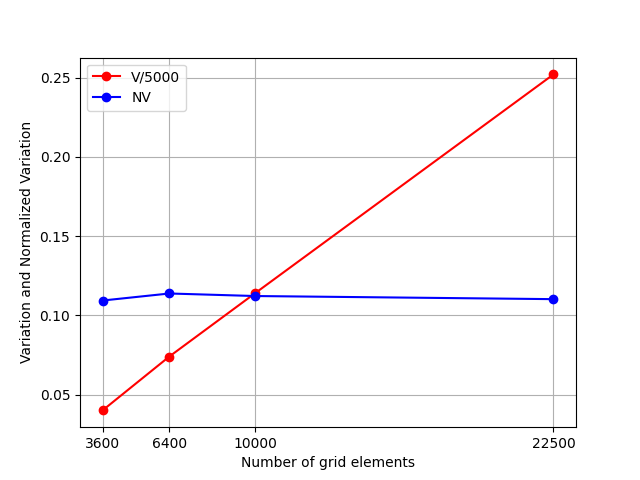}
    \caption{Variation and normalized variation of local entropy for a $2$D GCM model discussed later in Sec.~\ref{sec:2dGCM}. Variation values were divided by $5000$ for better visual comparison with normalized variation.}
    \label{fig:V_and_NV}
\end{figure}

\subsection{Chaotic itinerancy assessment}
\label{sec:assess}

We use the quantities defined in Sec. \ref{sec:matrices}--\ref{sec:var} for the assessment of whether a given dynamical system exhibits the phenomenon of chaotic itinerancy (CI) as described below. A schematic overview of the proposed numerical procedure is shown in Fig.~\ref{fig: assessment algorithm}.

\begin{figure}[htbp]
    \centering
    \includegraphics[width=\linewidth]{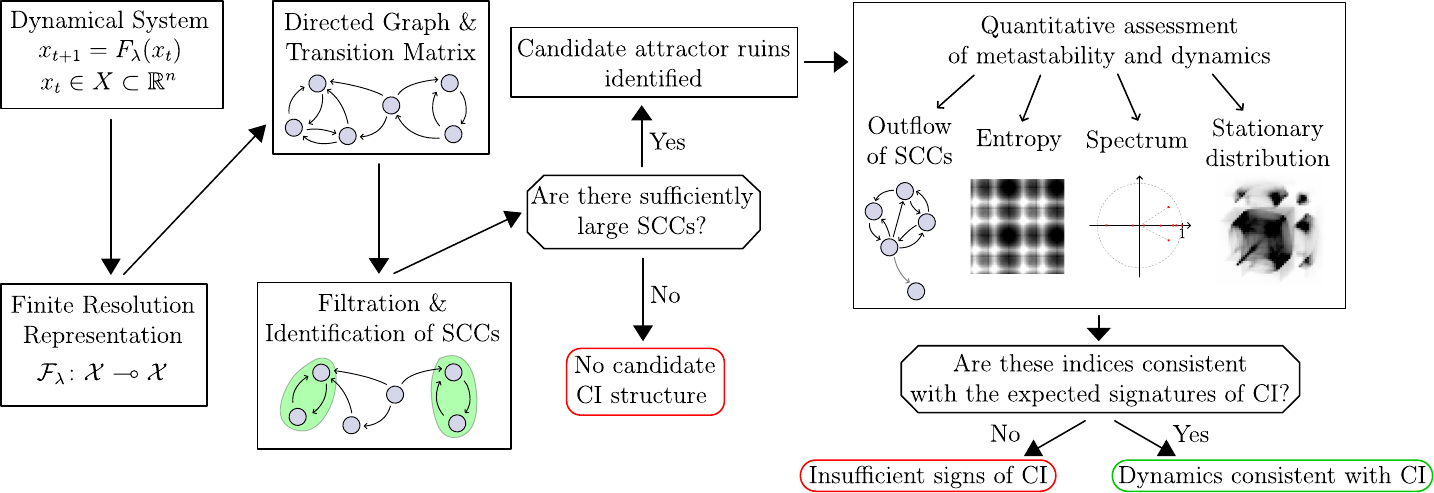}
    \caption{Schematic overview of the proposed procedure used to assess chaotic itinerancy for a given dynamical system.}
    \label{fig: assessment algorithm}
\end{figure}


Given a representation $\mathcal{F}$ of a map $f \colon X \to X$ with respect to a grid $\mathcal{X}$ in $X$, as explained in Sec.~\ref{sec:sets}, we consider the corresponding transition matrix $P$ and compute its stationary distribution (Sec.~\ref{sec:matrices}), spectrum (Sec.~\ref{sec:spectralgap}), and normalized local entropy and entropy rate (Sec.~\ref{sec:entropy}). Then we select a few different values of $\tau$, for example, $\tau \in \{0.2, 0.5, 0.8\}$, and we compute the corresponding collections of SCCs along with their outflow (Sec.~\ref{sec:SCC}). We also analyze the variation of the stationary distribution and entropy rate across the phase space (Sec.~\ref{sec:var}).

In a system that exhibits CI, the computed stationary distribution should reveal several regions of high density and a considerable amount of space with low density. In fact, a high variation of the density would be an indicator of the existence of such a structure.

Another indicator of CI is the coexistence of two or more eigenvalues close to $1$ and thus a small spectral gap. This coincides with the fact that we assume the existence of two or more metastable attractor ruins in chaotically itinerant dynamics. These attractor ruins correspond to strongly connected path components in the graph $\mathcal{F}_\tau$ for some choice of $\tau$.

To say that a given system exibits CI, we require that for $\tau=0$ the system is transitive on the global attractor, which is equivalent to saying that the graph $\mathcal{F}$ has a single non-trivial SCC that covers all the space where the long-term dynamics of the system takes place. Then for some positive value of $\tau$, the graph $\mathcal{F}_\tau$ should split into a collection of smaller SCCs that correspond to attractor ruins with low outflow and transient states characterized by outflow close to $1$. We also require SCCs corresponding to attractor ruins to have similar sizes (measured as the number of cubes in the combinatorial representation $\mathcal{F}$), accumulate similar amounts of $\pi$ within them, and we also require that the entropy rate restricted to those components is significantly lower than the entropy rate computed for the entire grid. To measure these values, we use the following auxiliary piecewise linear function $\mu \colon [0,1] \to [0,1]$ such that $\mu(0)=\mu(1)=0$ and $\mu(1/N)=1$, where $N$ is the number of SCCs found:
\begin{equation}\label{eq:func}
\mu(x) = \begin{cases}
    Nx, \quad\quad \quad\quad \ \  \text{for } x \le \frac{1}{N} \\
    \frac{N}{N-1}(1-x), \quad \text{otherwise}
\end{cases}
\end{equation}
For $j=1,2\ldots,N$, we compute $A_j = \mu(\frac{|\SCC_j|} {\sum_{i=1}^{N}|\SCC_i|})$ to quantify the relative sizes of SCCs. This kind of quantification produces a number in the range $[0,1]$, where $1$ is obtained if the size of the SCC equals the mean size of all the SCCs, and a number close to $0$ is obtained if the SCC is either very small or very large in comparison to the mean size of the SCCs. Similarly, for measuring accumulated $\pi$ within an SCC, we calculate $B_j = \mu(\frac{\pi(\SCC_j)} {\sum_{i=1}^{N}\pi(\SCC_i)})$.

\section{Computations for a selection of different models}
\label{sec:models}

To show the usefulness of our approach for the detection of chaotic itinerancy (CI) using the graph representation of a dynamical system introduced in Sec.~\ref{sec:sets} and the quantities and criteria introduced in Sec.~\ref{sec:metrics}, we first analyze a $2$-dimensional system that does not exhibit chaotic itinerancy (Sec.~\ref{sec:2dGCM}) and then a $3$-dimensional system in which we find chaotic itinerancy for some ranges of parameters (Sec.~\ref{sec:tsuda3d} and \ref{sec:compRange}).

\subsection{Two-dimensional coupled system of logistic maps}
\label{sec:2dGCM}

We begin with the system of two globally coupled logistic maps. In \eqref{eq:CML}, we set $N=2$, $a=2$ and $f(x)=1-ax^2$. Although this system does not, in fact, exhibit CI, it helps to familiarize oneself with the methods used and shows how to distinguish ``classic'' chaos from chaotically itinerant dynamics.

\subsubsection{Numerical simulations}
Our numerical simulations show that for different values of $\varepsilon\in \{0.15, 0.22, 0.23, 0.25\}$ one observes visually different dynamical behavior of trajectories in the GCM system. The map $f$ itself is chaotic, and we only play with the coupling strength $\varepsilon$. For $\varepsilon=0.15$ and $\varepsilon=0.22$, the chaotic behavior of $f$ dominates the coupling. In the case of $\varepsilon=0.23$ and $\varepsilon=0.25$, trajectories settle on a chaotic attractor after an initial period of chaotic wandering. Sample trajectories for these cases are shown in Fig.~\ref{fig:2dGCM_sample_trajectories}. In the remainder of this subsection, we shall focus on these four chosen values of the parameter~$\varepsilon$.

\begin{figure}[htbp]
    \centering

    \begin{subfigure}{0.24\textwidth}
    \includegraphics[width=\linewidth]{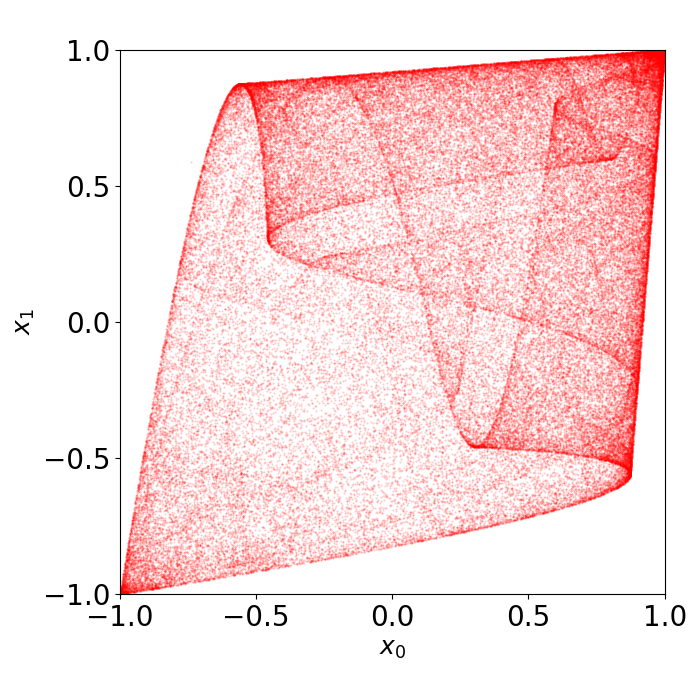}
    \caption{}\label{fig:2dGCMtrajectory_eps_015}
    \end{subfigure}
    \begin{subfigure}{0.24\textwidth}
    \includegraphics[width=\linewidth]{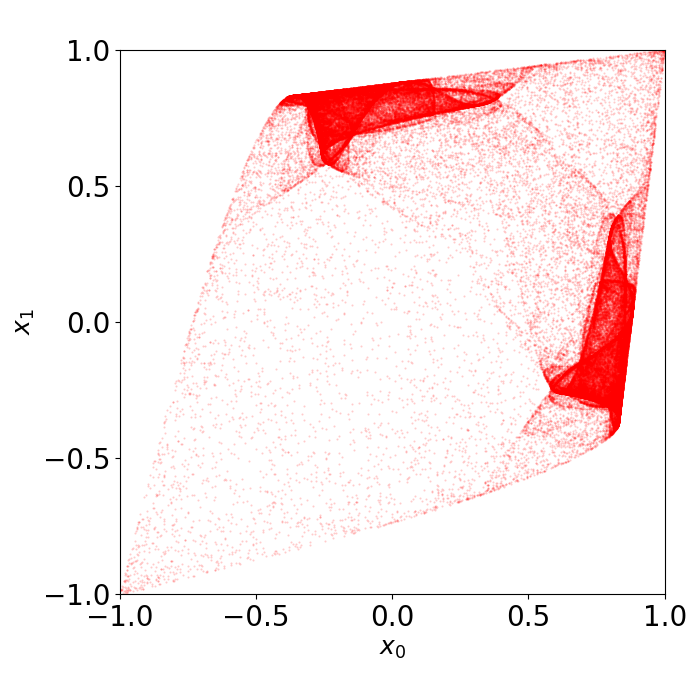}
    \caption{}\label{fig:2dGCMtrajectory_eps_022}
    \end{subfigure}
    \begin{subfigure}{0.24\textwidth}
    \includegraphics[width=\linewidth]{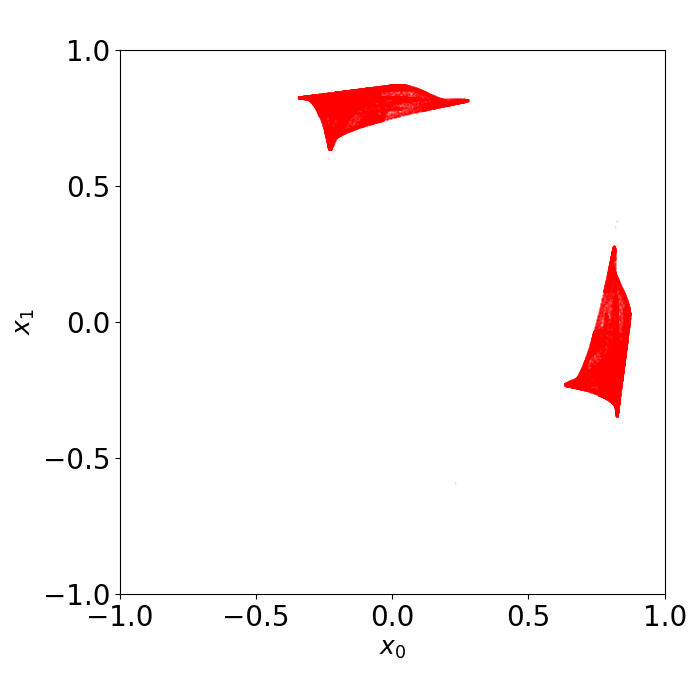}
    \caption{}\label{fig:2dGCMtrajectory_eps_023}
    \end{subfigure}
    \begin{subfigure}{0.24\textwidth}
    \includegraphics[width=\linewidth]{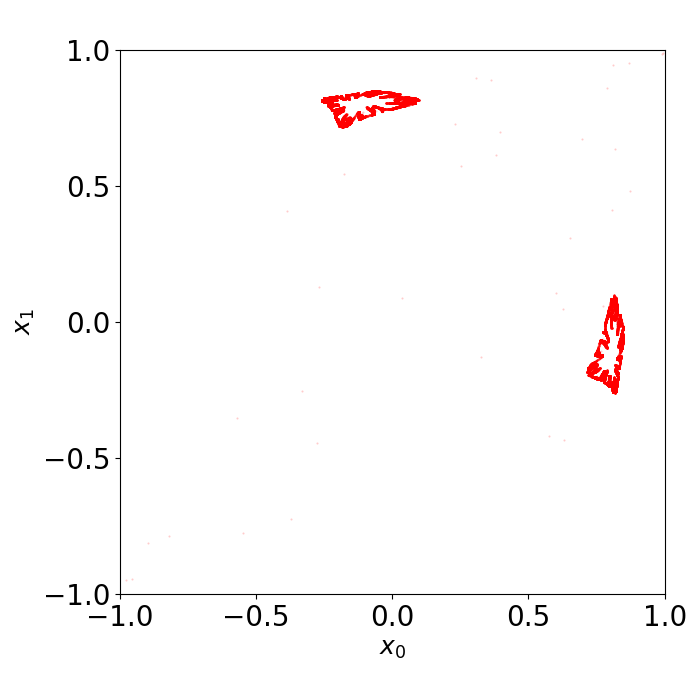}
    \caption{}\label{fig:2dGCMtrajectory_eps_025}
    \end{subfigure}
    
    \caption{Sample trajectories of the $2$-dimensional GCM system discussed in Sec.~\ref{sec:2dGCM} for a few values of $\varepsilon$: \mbox{(A) $\varepsilon=0.15$,} (B) $\varepsilon=0.22$, (C) $\varepsilon=0.23$, (D) $\varepsilon=0.25$}
    \label{fig:2dGCM_sample_trajectories}
\end{figure}

\subsubsection{Discretization of the system}
For the purpose of using our proposed framework, we shall work with the state space $X = [-1,1] \times [-1,1]$ discretized into the grid of $200 \times 200$ squares ($2$-dimensional cubes) of the same size. For the Monte Carlo simulation step, we take $250$ samples from each square. The computed digraph has $40\,000$ vertices in each case. These vertices represent all the squares that together form the grid.

\subsubsection{Stationary distribution and local entropy}
We begin by computing a stationary distribution and the local entropy, and then by quantifying their variation using the formula \eqref{eq:NormVariation}. Outcome of such computations is shown in Table \ref{tab:2dCGMvar}. In each case we see that the normalized variation of local entropy, $\NV(h_i, \mathcal{X})$, is much smaller than the amplitude $\text{amp}(h_i)$ across the whole grid. The opposite holds true for the stationary distribution. In cases $\varepsilon=0.15$ and $\varepsilon=0.22$ where chaos is present, $\NV(\pi, \mathcal{X})$ is larger than $\text{amp}(\pi)$, and in cases where a typical trajectory settles on a chaotic attractor,  the opposite occurs.

\begin{table}[htbp]
\centering
\renewcommand{\arraystretch}{1.3}
\setlength{\tabcolsep}{9pt}
\begin{tabular}{c | c | c | c | c} 
$\varepsilon$  & $0.15$ & $0.22$ & $0.23$ & $0.25$  \\
\hline
$\text{amp}(h_i)$ & 0.52903 & 0.52557 & 0.51943 & 0.51962 \\ 
$\NV(h_i, \mathcal{X})$ & 0.10196 & 0.10545 & 0.10668 & 0.10608 \\ 
$\text{amp}(\pi)$ & 0.00259 & 0.00218 & 0.00470 & 0.00629 \\
$\NV(\pi, \mathcal{X})$ & 0.00307 & 0.00497 & 0.00300 & 0.00311 \\
$\hat{h}$ & 0.32108 & 0.28572 & 0.27166 &  0.26434 \\
$\gamma$ & 0.34175 & 0.01968 & 0.00176 & 0.00001 \\
$||\pi P - \pi||_{l_1}$ & 0.00078 & 0.00140 & 0.00176 & 0.00190
\end{tabular}
\caption{Normalized variation of local entropy, stationary distribution, entropy rate, spectral gap and error of estimation of $\pi$ for the considered values of $\varepsilon$ in the 2-dimensional GCM system.}
\label{tab:2dCGMvar}
\end{table}

\subsubsection{Detection of attractor ruins}
In the next step, we search for strongly connected components (SCCs) across many thresholds of filtration to find collections of co-existing attractor ruins.
In the cases of $\varepsilon=0.15$, $\varepsilon=0.22$ and $\varepsilon=0.23$, we do not observe collections of two or more SCCs approximately equal in size for any choice of $\tau$; therefore, these cases do not constitute valid candidates for chaotic itinerancy. The sizes of the SCCs found for $\varepsilon = 0.25$ are shown in Fig.~\ref{fig:sizes_of_SCCs_of2dGCM} along with the value of $\tau$ for which they arise.
\begin{figure}[tbp]
    \includegraphics[width=0.8\linewidth]{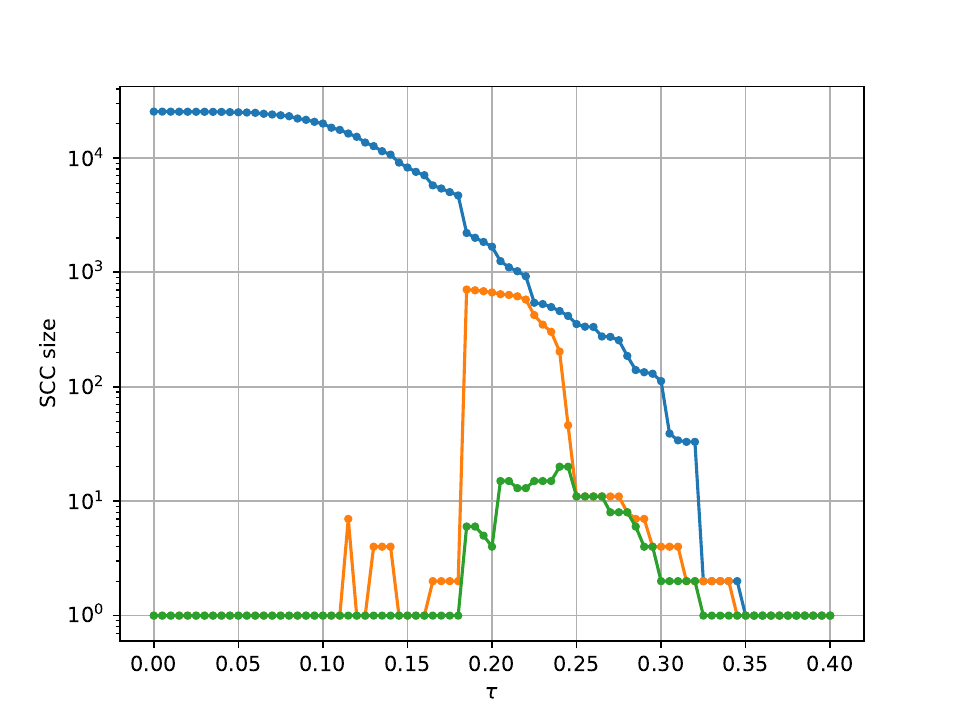}
    \label{fig:2dGCMeps025size}
    \caption{Sizes of the largest strongly connected components within $\mathcal{F_\tau}$ for $\tau \in [0,0.4]$ with step $0.005$ and $\varepsilon=0.25$. Note the logarithmic scale on vertical axis.}
    \label{fig:sizes_of_SCCs_of2dGCM}
\end{figure}
The two SCCs of similar size found for the filtration threshold $\tau=0.185$ are shown in Fig.~\ref{fig:SCCS_for2dGCMeps025}.
\begin{figure}[htbp]
    \centering
    \includegraphics[width=0.8\linewidth]{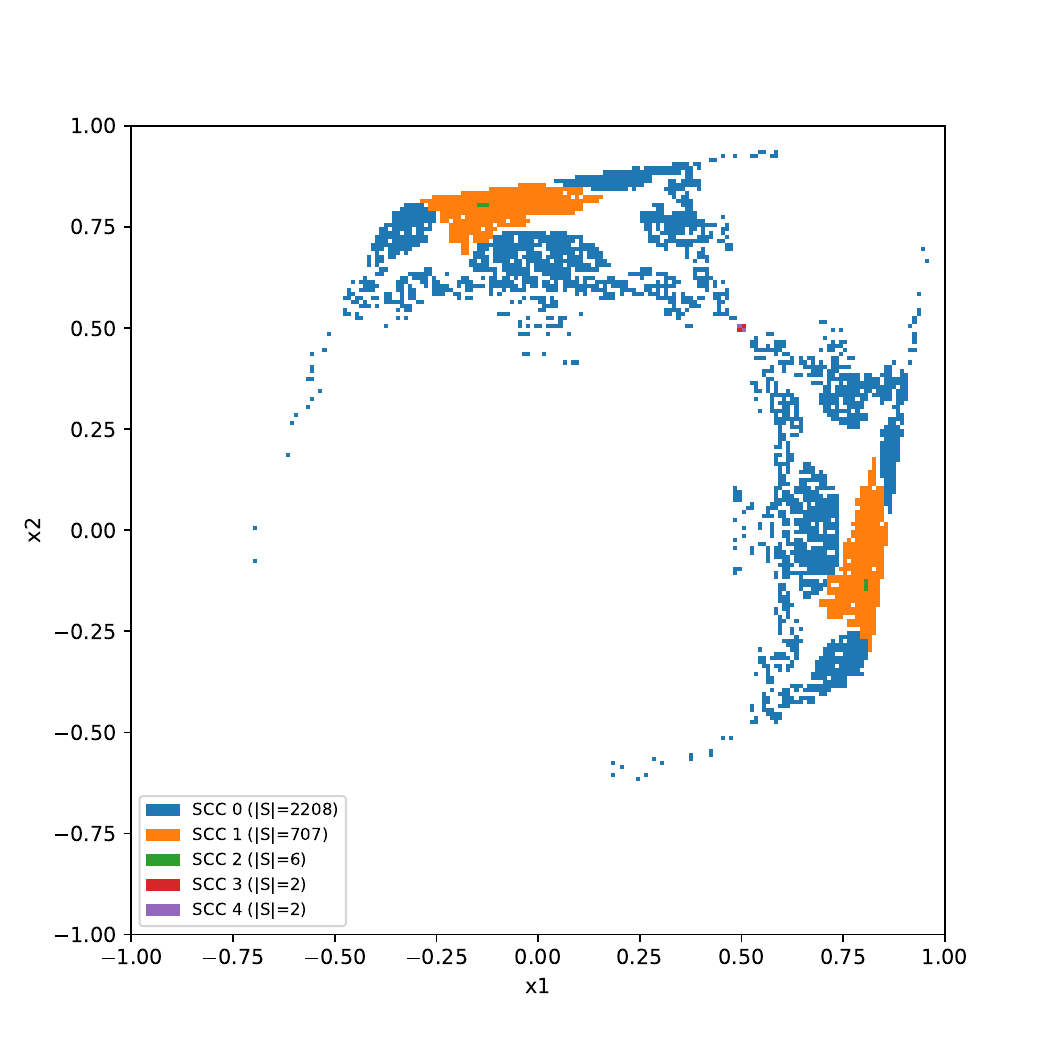}
    \caption{Two large and three small SCCs found for $\varepsilon=0.25$ and filtration threshold $\tau=0.185$.}
    \label{fig:SCCS_for2dGCMeps025}
\end{figure}
The sizes of these SCCs are $2208$ and $707$, and their respective outflows are approximately $0.21$ and $0.05$. The first of those SCCs represents the chaotic state of the system.
The second SCC covers the attractor of the system (shown in Fig.~\ref{fig:SCCS_for2dGCMeps025}) and, due to the finite discretization of state space, a small portion of edges point out of the set.
The other SCCs seem to be transient states, with a very small size and outflow close to $1$. The three eigenvalues largest in magnitude for $\varepsilon=0.25$ are $\{1, -0.999, 0.9603\}$; this suggests strong, almost persistent oscillations, and is consistent with the observation of typical trajectories for this system.

\subsubsection{Assessment of chaotic itinerancy}
The final step is to compute $A_j$ and $B_j$ for the SCCs found at the chosen threshold; these are the quantities defined in Sec.~\ref{sec:assess} to check how well balanced the sizes of the SCCs are. Those indices, as well as normalized entropy rates and outflow, are gathered in Table~\ref{tab:ABCDkappa}. The small values of most of the indices rule out the possibility of CI within this system.

\begin{table}[htbp]
\centering
\renewcommand{\arraystretch}{1.3}
\setlength{\tabcolsep}{9pt}
\begin{tabular}{c | c | c}
 & \textbf{$\SCC_1$} & \textbf{$\SCC_2$} \\
\hline
$A_j$ & 0.48507 & 0.48507 \\
$B_j$ & 0.01585 & 0.01864 \\
$\Phi_j$ & 0.21343 & 0.05139 \\
$\hat{h}(\SCC_j)$ & 0.28402 & 0.26317 \\
\end{tabular}
\caption{Indices computed for the $2$-dimensional GCM system with $\varepsilon=0.25$.}
\label{tab:ABCDkappa}
\end{table}

\subsection{Three-dimensional coupled map lattice}\label{sec:tsuda3d} 
To follow up with a more complex system that shows signs of CI under visual inspection, we carry out computation for model \eqref{eq:tsudaCML}. Figures for this particular model were included throughout this article and now we conclude with discussing the outcome of extensive computations.

\subsubsection{Discretization of the system}
The whole dynamics of the system is contained within $[-0.5,1.5]^3\subseteq\mathbb{R}^3$, and thus we choose this region as the domain $X$.
The region $X$ is then divided into a collection of $216\ 000$ cubes ($60\times60\times60$) that together form a grid $\mathcal{X}$, which is a basis for the definition of the combinatorial map $\mathcal{F}$. We set $250$ as the number of samples for the Monte Carlo simulation step.

\subsubsection{Stationary distribution and local entropy}
Table~\ref{tab:TsudaCMLvar} shows the variation of the local entropy and the stationary distribution, their respective amplitudes, the entropy rate, the spectral gap, and the error of the approximation of $\pi$.

\begin{table}[htbp]
\centering
\renewcommand{\arraystretch}{1.3}
\setlength{\tabcolsep}{9pt}

 \begin{tabular}{c | c} 
 quantity & value \\
  \hline
 $\text{amp}(h_i)$ & 0.79388 \\ 
 $\NV(h_i, \mathcal{X})$ & 0.12483 \\ 
 $\text{amp}(\pi)$ & 0.00215 \\
 $\NV(\pi, \mathcal{X})$ & 0.00577 \\
 $\hat{h}$ & 0.48386 \\
 $\gamma$ & 0.10504 \\
 $||\pi P - \pi||_{l_1}$ & 0.00158 \\ [1ex] 
 \end{tabular}
 \caption{Normalized variation of local entropy, stationary distribution, entropy rate, spectral gap and error of estimation of $\pi$ for the considered model \eqref{eq:tsudaCML}.}
 \label{tab:TsudaCMLvar}
\end{table}

\subsubsection{Detection of attractor ruins}
We work with the filtration of the graph $\mathcal{F}$ in order to capture the SCCs that appear for the threshold $\tau \in [0.12, 0.2]$, with computations made with step $0.001$. As we see in Fig.~\ref{fig:SCC_sizes_tsuda}, for $\tau\in[0.12, 0.132]$, there is only one global SCC that contains all the attractor ruins together with chaotic transitions between them. Then for $\tau = 0.133$, enough light edges in $\mathcal{F}$ are cut off, and therefore we see the emergence of another SCC; our computations show that its outflow is low. This process repeats for $\tau=0.137$, uncovering another SCC that corresponds to an attractor ruin in the state space, as shown in Fig.~\ref{fig:SCCs_in_state_space_CML}.

\subsubsection{Assessment of chaotic itinerancy}
The indices computed for the SCCs obtained for $\mathcal{F}_{0.137}$ are gathered in Table~\ref{tab:TsudaCMLABCD} and indicate that this model can be assessed as a case of chaotic itinerancy. The values of $A_j$ and $B_j$ are close to $1$ in the three cases, and $\Phi_j$ (the outflow of the $j$-th component) takes values below $0.5$. Note that entropy rates restricted to the SCCs are by about $30\%$ lower than the global entropy rate $\hat{h} \approx 0.48$.

\begin{figure}[htbp]
    \centering
    \includegraphics[width=0.8\linewidth]{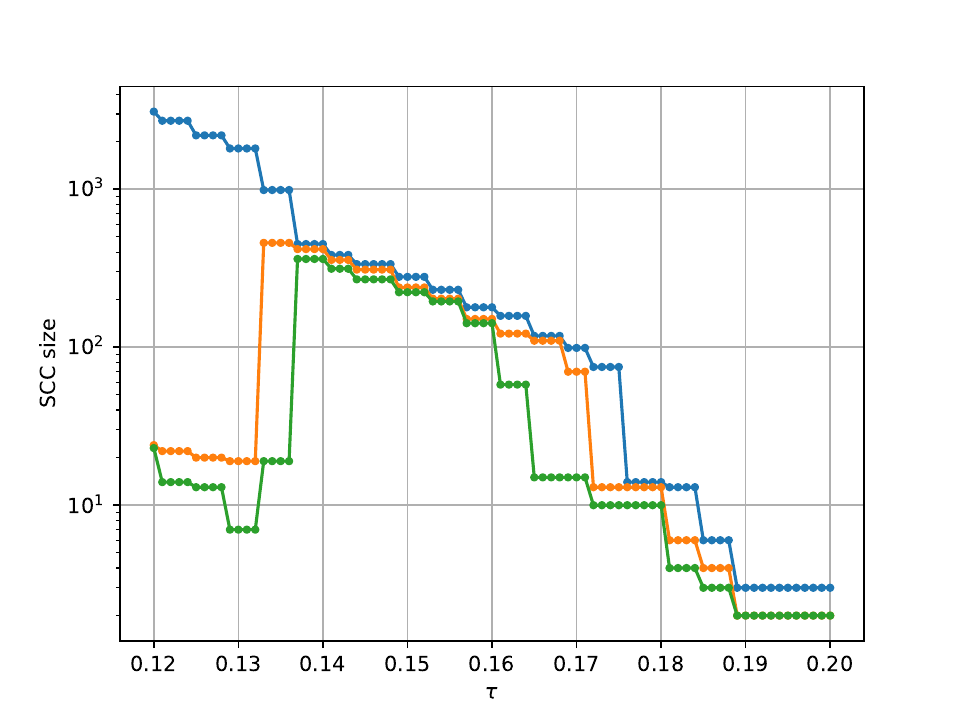}
    \caption{Sizes of the three largest SCCs as a function of threshold $\tau\in[0.12, 0.2]$ of the filtration for the $3$-dimensional CML system~\eqref{eq:tsudaCML}.}
    \label{fig:SCC_sizes_tsuda}
\end{figure}

\begin{figure}[htbp]
    \begin{subfigure}{0.48\textwidth}
    \includegraphics[width=\linewidth]{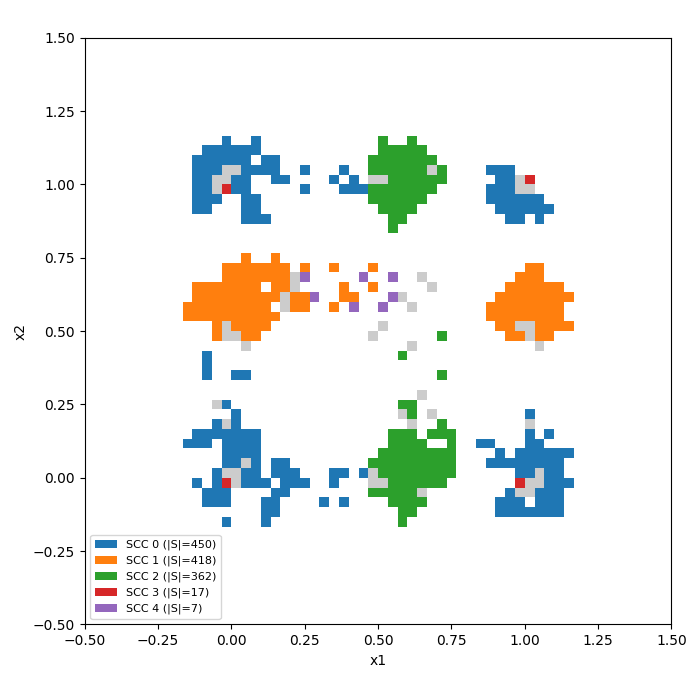}
    \caption{}\label{fig:tsudaSCCs_2dview_tau138}
    \end{subfigure}
    \begin{subfigure}{0.48\textwidth}
    \includegraphics[width=\linewidth]{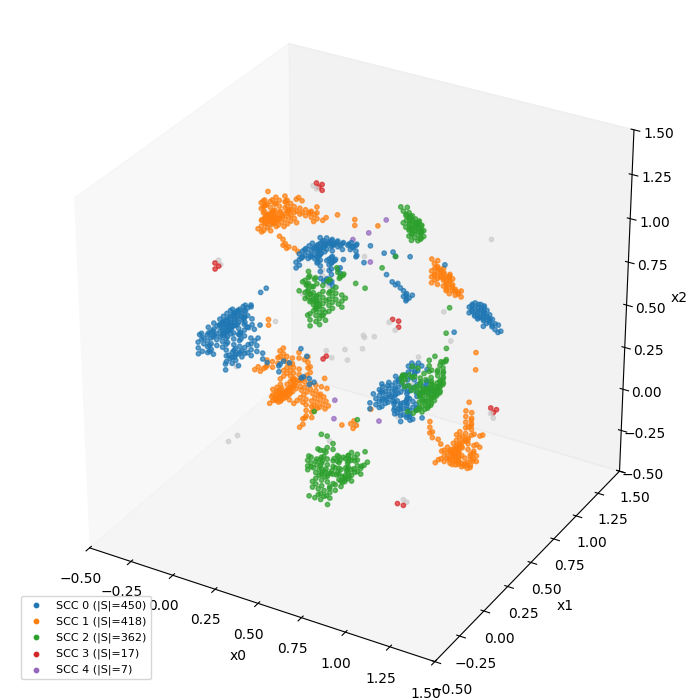}
    \caption{}\label{fig:tsudaSCCs_3dview_tau138}
    \end{subfigure} 
    \caption{SCCs on the grid $\mathcal{X}$ projected onto the first two coordinates for $\tau=0.137$ (A), and centers of the cubes shown as a $3$-dimensional plot (B).}
    \label{fig:SCCs_in_state_space_CML}
\end{figure}

\begin{table}[htbp]
\centering
\renewcommand{\arraystretch}{1.3}
\setlength{\tabcolsep}{9pt}

\begin{tabular}{c | c | c | c}
 & \textbf{$\SCC_1$} & \textbf{$\SCC_2$} & \textbf{$\SCC_3$} \\
\hline
$A_j$ & 0.95121 & 0.99024 & 0.88292 \\
$B_j$ & 0.97644 & 0.99490 & 0.94270 \\
$\Phi_j$ & 0.42089 & 0.44359 & 0.32364 \\
$h(\SCC_j)$ & 0.35631 & 0.35483 & 0.35328 \\
\end{tabular}
\caption{Indices computed for the first three SCCs in the $3$D CML system \eqref{eq:tsudaCML}.}
\label{tab:TsudaCMLABCD}
\end{table}

\subsection{Computations for a range of parameters.}
\label{sec:compRange}

We conclude the paper by describing a comprehensive computation that we conducted for a range of parameters in the $3$-dimensional CML system \eqref{eq:tsudaCML}. We considered the parameter space $\Lambda = [0.25,0.75]\times[0.05,0.15]$, from which we sampled $41\times 51$ evenly spaced pairs $(\omega,\varepsilon) \in \Lambda$, obtaining $2091$ pairs in total. For each pair, we constructed grid $\mathcal{X}$ consisting of $216\ 000$ cubes and numerically approximated the combinatorial map $\mathcal{F}_{\omega,\varepsilon}\colon \mathcal{X}\multimap \mathcal{X}$ as described in sections \ref{sec:sets} and \ref{sec:tsuda3d}. For each map, the collection of indices described in Sec.~\ref{sec:assess} were computed and combined into a database consisting of results obtained for all the pairs of parameters. Let us illustrate and discuss the outcome of these computations.

\subsubsection{Search for attractor ruins}
We first check the presence of at least three strongly connected components composed of $200$ or more cubes for some value of the filtration threshold $\tau$. The overview of the parameter space $\Lambda$ and pairs $(\omega,\varepsilon)$ that fulfill this requirement is shown in Fig.~\ref{fig: tau-value at the moment of SCC emergence on parameter space}.
\begin{figure}[htbp]
    \centering
    \includegraphics[width=0.75\linewidth]{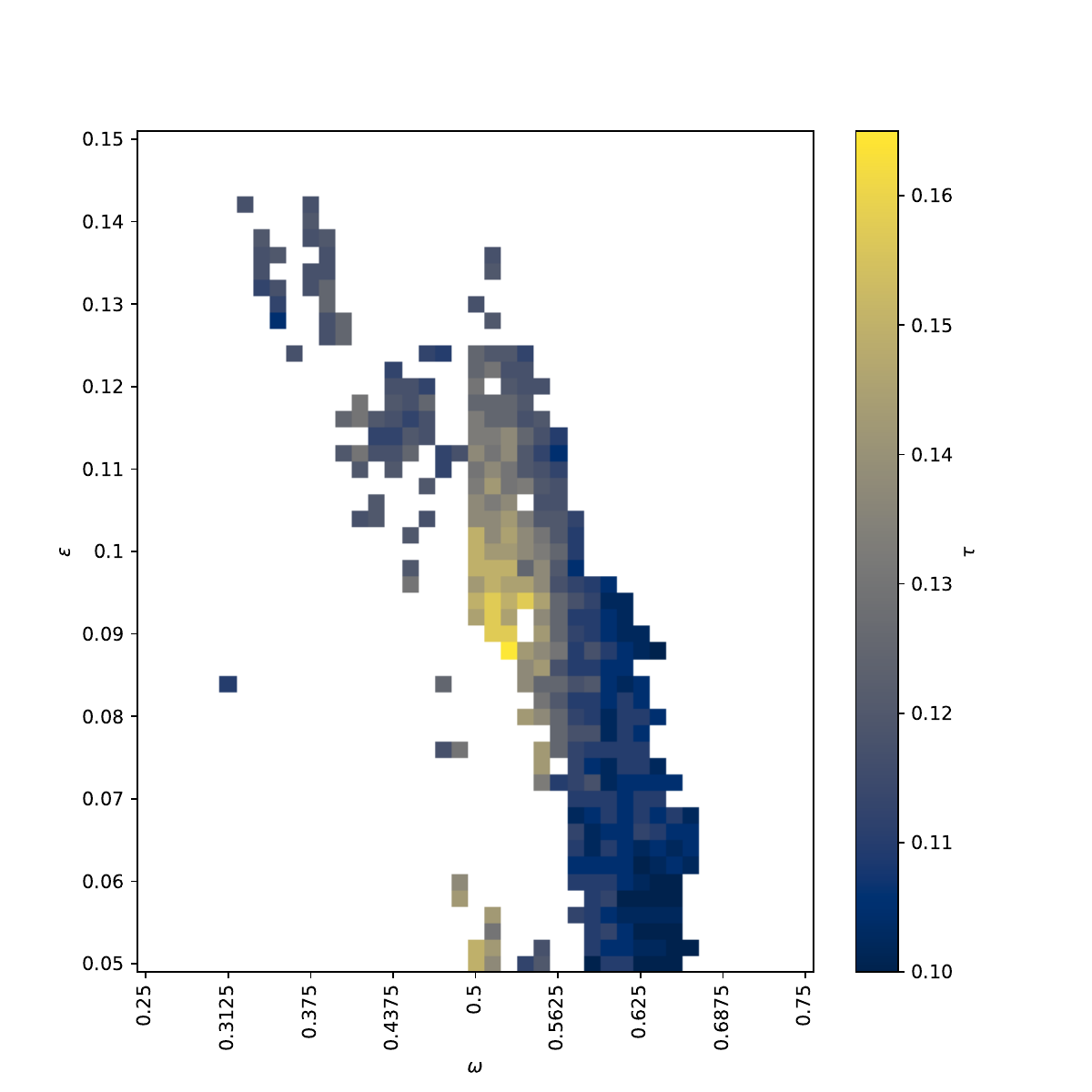}
    \caption{The minimum value of $\tau$ in the filtration for which at least three SCCs composed of $200$ or more cubes are present.}
    \label{fig: tau-value at the moment of SCC emergence on parameter space}
\end{figure}
The color used in this figure indicates the lowest value of $\tau$ that results in the emergence of such components, while the white color indicates parameter pairs that do not admit three SCCs of sufficient size for all values of $\tau$ considered. For the remainder of this section, we shall focus only on those parameter pairs that do admit such components.

One can notice that parameters around $(\omega,\varepsilon) \approx (0.55,0.09)$ yield the highest values of $\tau$, which can be interpreted as strong connections between ruin attractors that correspond to relatively many trajectories running from one attractor ruin to the other.

\subsubsection{Stationary distribution and its variation}
In the next step, we compute an approximation of the stationary distribution $\pi$ for each $\mathcal{F}_{\omega,\varepsilon}$. Then we compute its normalized variation $\NV(\pi, \mathcal{X})$. This quantity for the different parameters is shown in Fig.~\ref{fig: NVpi on parameter space}.

\begin{figure}[htbp]
    \centering
    \includegraphics[width=0.75\linewidth]{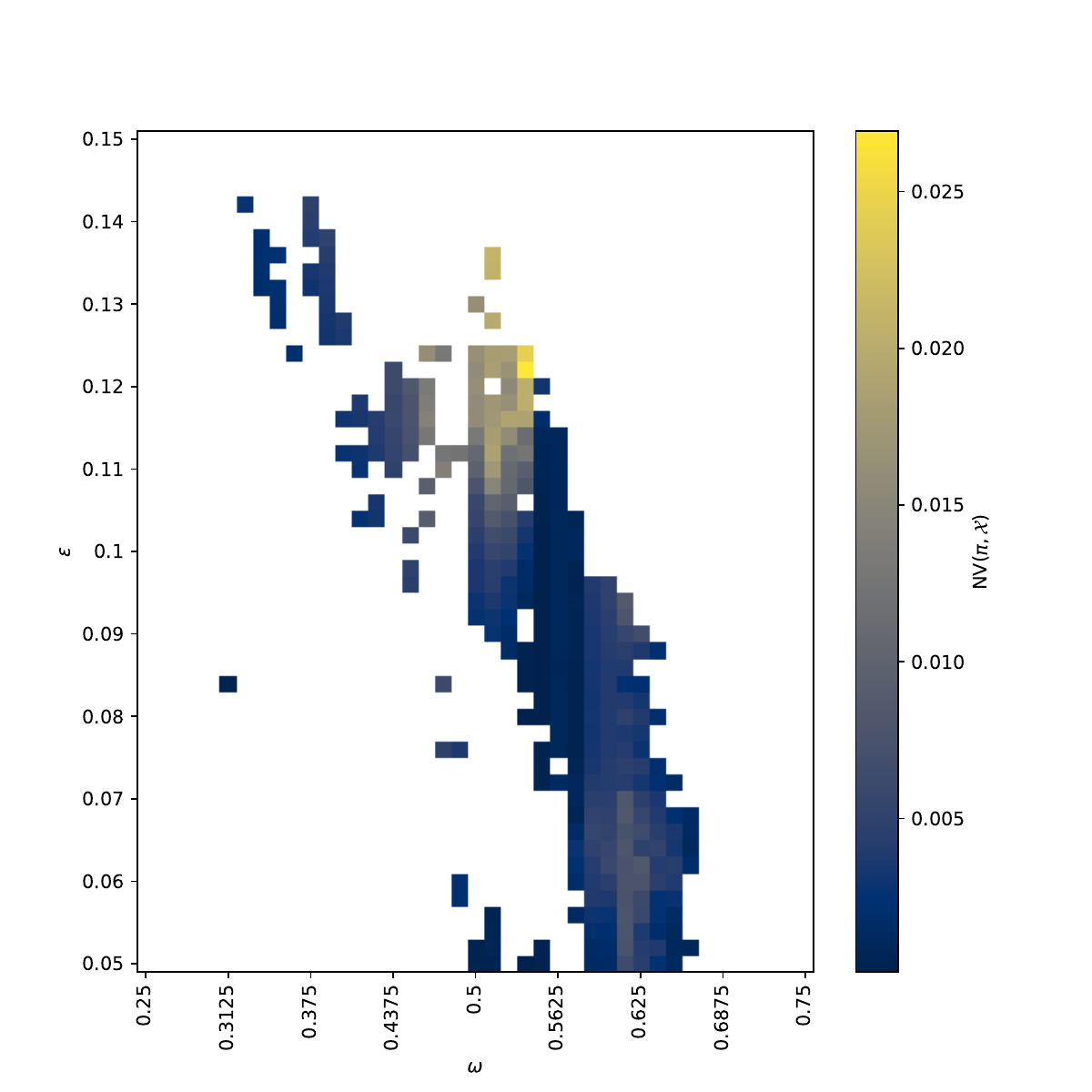}
    \caption{Normalized variation of the computed stationary distribution $\pi$ across those parameter pairs for which at least three SCCs of $200$ or more cubes are present.}
    \label{fig: NVpi on parameter space}
\end{figure}

One can see that the highest values of $\NV(\pi, \mathcal{X})$ can be found for approximately $(\omega, \varepsilon)\in [ 0.5,0.55]\times[0.104,0.126]$. The other pairs of parameters yield considerably lower variation.

\subsubsection{Spectral gap}
Having computed the stationary distribution, we compute the eigenvalues of the transition matrix and its spectral gap. The result is shown in Fig.~\ref{fig: spectral gap on parameter space}.

\begin{figure}[htbp]
    \centering
    \includegraphics[width=0.75\linewidth]{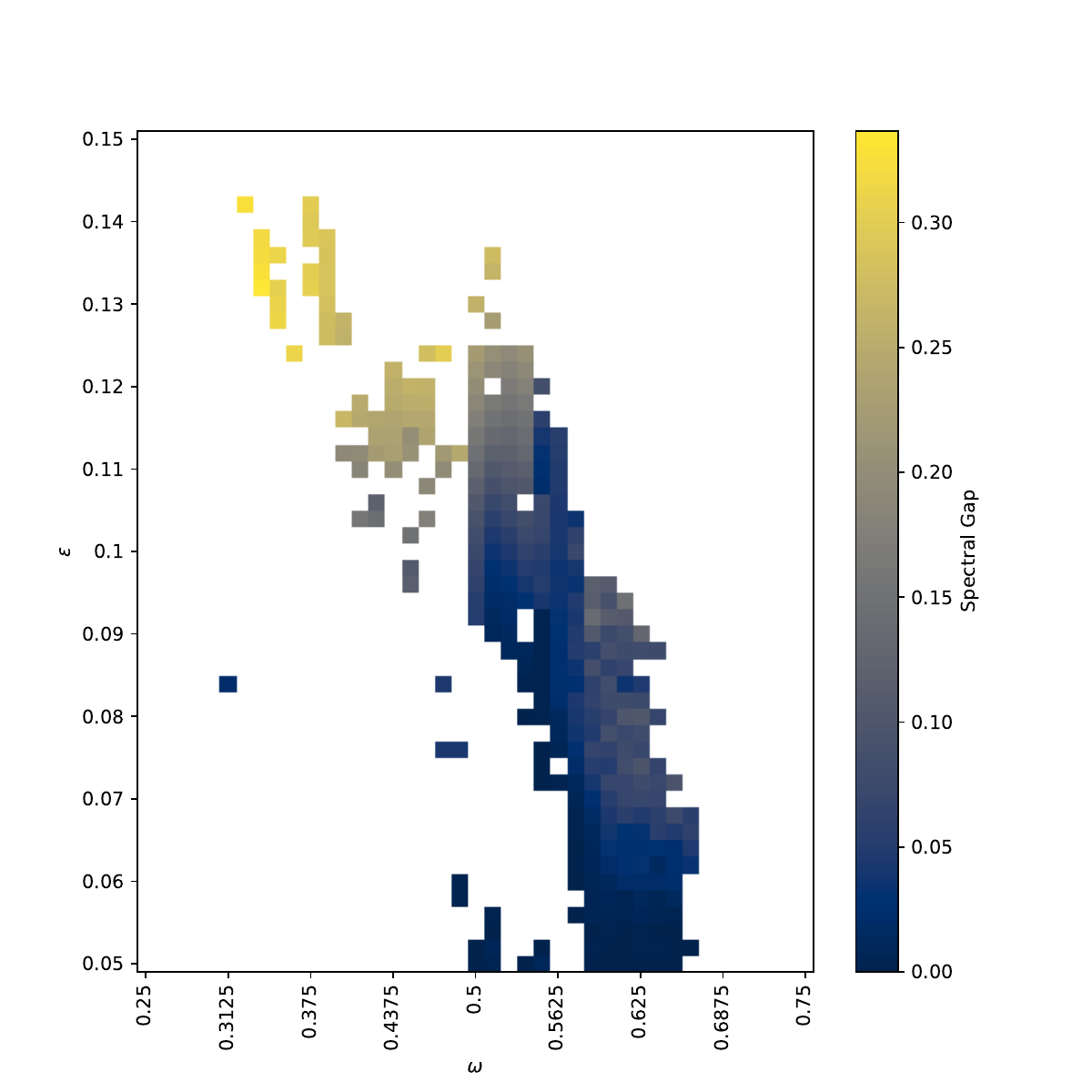}
    \caption{The spectral gap of $\mathcal{F}$ for parameter pairs for which at least three SCCs composed of $200$ or more cubes are present.}
    \label{fig: spectral gap on parameter space}
\end{figure}

One can notice two characteristic regions in the parameter space: values for $(\omega, \varepsilon)$ where $\omega<0.5$ and $\varepsilon > 0.1$ generally yield higher spectral gap than those for $\omega \ge 0.5$.
In the first case, the high value of the spectral gap indicates fast stabilization and the lack of bottleneck clusters in $\mathcal{F}$. On the contrary, the low value of the spectral gap in the second case suggests that for those pairs of parameters, multiple metastable clusters might be present in $\mathcal{F}$.

\subsubsection{Entropy analysis}
Figure \ref{fig: entropy-based index on parameter space} shows the entropy rate encountered across the parameter space $\Lambda$.
\begin{figure}[hbp]
    \centering
    \includegraphics[width=0.75\linewidth]{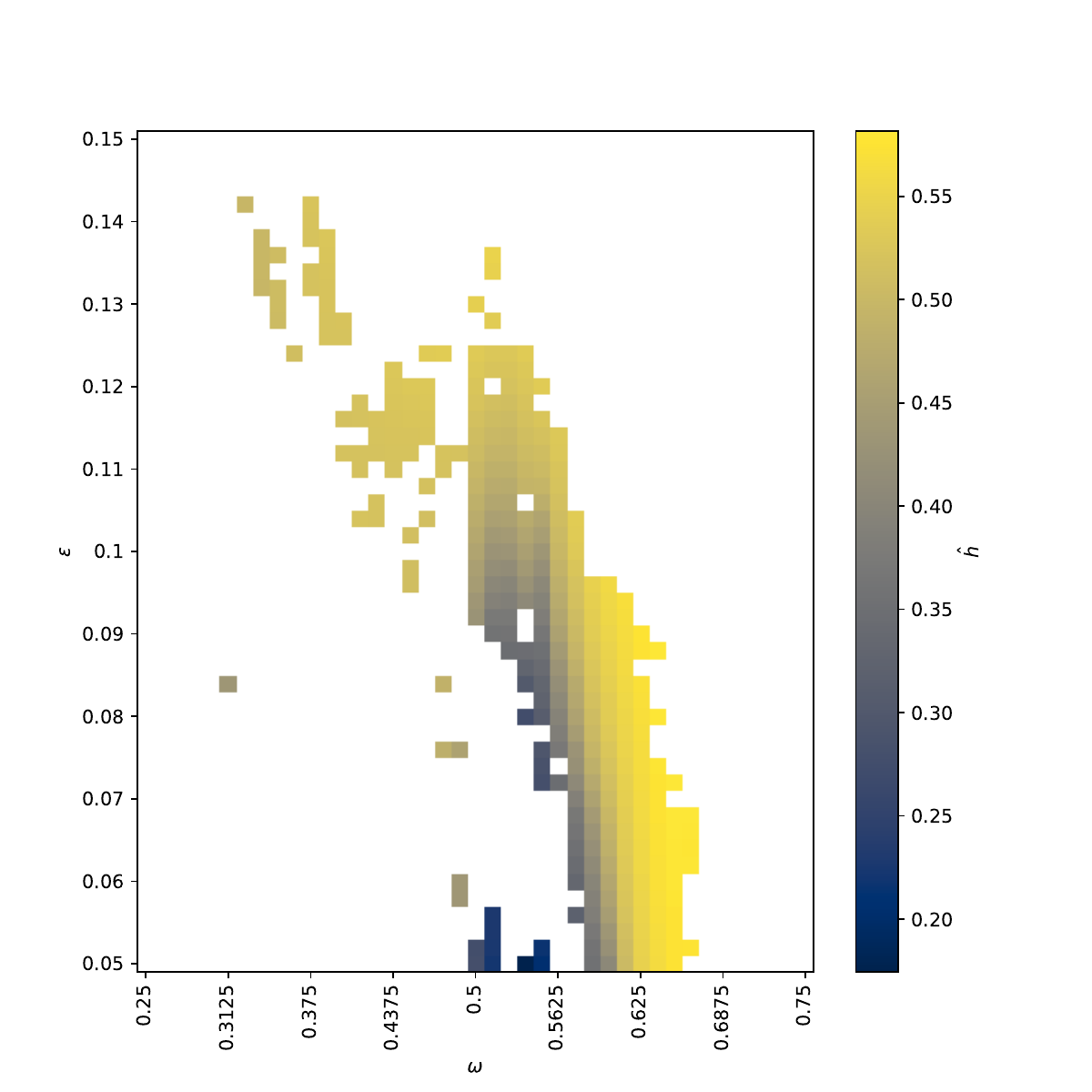}
    \caption{Entropy rate.}
    \label{fig: entropy-based index on parameter space}
\end{figure}
One can distinguish three regions here: values for $(\omega, \varepsilon)$, where $\omega>0.6$, feature the highest entropy rate among all the considered parameter pairs, which together with previous observations suggests that the dynamics here is fully turbulent.
In the region of the considered part of the parameter space where $\varepsilon>0.11$, the entropy rate is still high but noticeably lower than in the first mentioned region. This part of $\Lambda$ essentially coincides with the region in which large spectral gap was found; see in Fig.~\ref{fig: spectral gap on parameter space}.
Entropy rate for pairs in which $\omega$ is approximately in $[0.5, 0.5625]$ and $\varepsilon<0.11$ have a relatively low entropy rate.

In order to compare the entropy rate restricted to SCCs with the global entropy rate, we carried out additional computations. In the majority of the considered parameter region, the entropy rate within SCCs was lower than the global one. It is consistent with the assumption that the dynamics within the attractor ruins should be more ordered. Figures for these results are not shown for brevity.

\subsubsection{Outflow of the SCCs}
The outflow of the SCCs obtained for the minimum value of $\tau$ is shown in Fig. \ref{fig: cond on parameter space}.
\begin{figure}[htbp]
    \centering
    \includegraphics[width=0.75\linewidth]{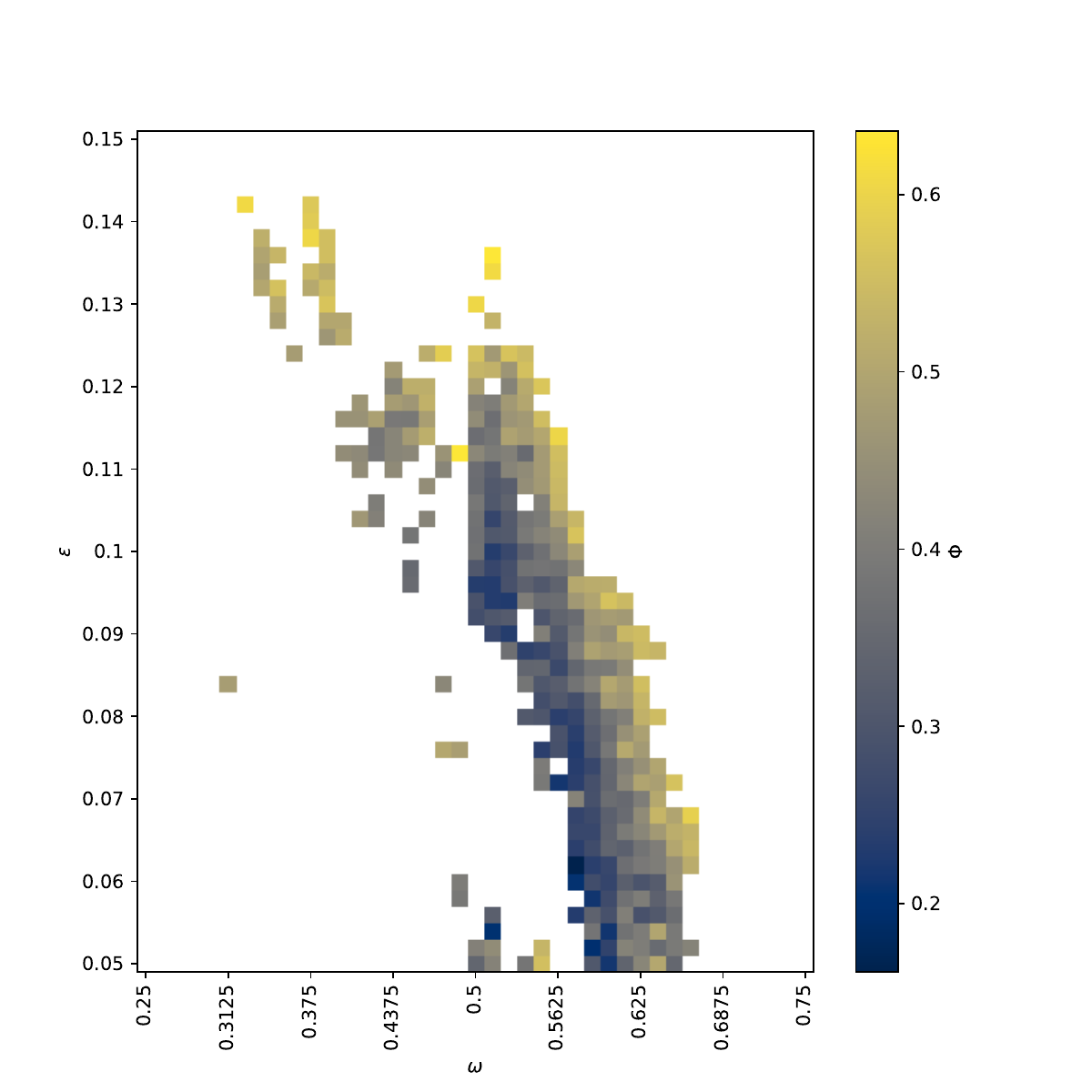}
    \caption{Maximum outflow among all the SCCs computed at their first appearance.}
    \label{fig: cond on parameter space}
\end{figure}
The value of each pixel here indicates the highest outflow among the relevant SCCs that are present in the corresponding graph. It therefore quantifies the instability of potential attractor ruins and and provides an upper bound on their outflow: all other SCCs corresponding to the attactor ruins have an outflow less than or equal to this value. Consequently, lower values indicate that all identified SCCs are better candidates for attractor ruins. 

\subsubsection{Assessment of chaotic itinerancy}
Taking into consideration all the results discussed above, we conclude that we found the most pronounced indications of chaotic itinerancy in the $3$-dimensional CML system \eqref{eq:tsudaCML} for parameters approximately given by $(\omega,\varepsilon)\in[0.5,0.55]\times [0.9,0.12]$. However, we must note that the dynamics of this model can change rapidly under small perturbations of parameters, while our computations were only carried out for a finite set of parameter pairs. Therefore, our results provide a coarse-grained approximate indicator of where chaotic itinerancy can be present across the analyzed parameter ranges.

\section{Conclusion and final remarks}
\label{sec:conclusion}
We have provided an effective computational framework for detecting the phenomenon of chaotic itinerancy in a dynamical system with discrete time on a compact subset of a Euclidean space, and we showed an application of this method to a specific 3-dimensional CML system and to a few other systems as well.

We have focused on a dynamical system given by a continuous map $f \colon X \to X$ defined by a known formula. However, our approach can also be applied to the analysis of a map $f$ given by a finite sample, that is, a finite set of pairs $(x_i,y_i)$ such that $y_i$ is supposedly an approximation of $f(x_i)$. Specifically, instead of using the Monte Carlo-like approach to create a graph representation $\mathcal{F}$ of $f$, we would use the pairs provided. Such a sampled map may come from time series data obtained, for example, in some physical or numerical experiment, or as a result of medical examination of a patient, such as ECG or EEG data. We emphasize the fact that, thanks to working with a filtration $\mathcal{F}_\tau$, our method is automatically resistant to a limited amount of noise, because it ignores assignments (graph edges) that are rarely followed.

Let us now point out some weaknesses of our approach and suggest some methods that could be used to overcome them.

Due to subdividing the phase space into a rectangular grid, our approach becomes inefficient in higher dimensions because the number of cubes to fill a region of a given size grows exponentially with the dimension; this is the so-called \emph{curse of dimensionality}. However, we expect that using coarse grids might yield a method that is also feasible in higher dimensions. Another improvement might focus on restricting the constructed grid to a subset of the phase space containing most of the dynamics; such a subset might be found by tracing a collection of forward trajectories of several initial points and determining a region of the phase space where high density of points is encountered. Moreover, instead of a rectangular grid, one could use a finite collection of balls that cover the global attractor of the system. Since attractors in high-dimensional systems are often low dimensional structures, one should expect to be able to construct a cover for such an attractor that consists of a relatively low number of balls. The ideas of persistent homology \cite{EdeHar2010} might be helpful here. Another idea might be to select a finite number of representative points in the phase space and to construct a Voronoi diagram instead. The different ideas proposed above may perform better in some cases than in others, and we think they are worth further investigation in the future.

\section*{Author contributions (CRediT)}

\textbf{Wojciech Jaworek:} Data curation, Formal analysis, Investigation, Methodology, Software, Validation, Visualization, Writing -- original draft, Writing -- review \& editing.
\textbf{Paweł Pilarczyk:} Conceptualization, Formal analysis, Funding acquisition, Investigation, Methodology, Project administration, Validation, Writing -- review \& editing.

\section*{Acknowledgments}

This research was supported by the National Science Centre, Poland, within the grant OPUS 2021/41/B/ST1/00405.
Computations were carried out using the computers of Centre of Informatics Tricity Academic Supercomputer \& Network.

P.\ Pilarczyk expresses his gratitude to Professor Hiroshi Kokubu (Kyoto University) for introducing him to the topic of chaotic itinerancy.

\section*{Conflict of Interest}
The authors have no conflicts to disclose.

\section*{Data availability statement}

All data shown and discussed in the paper can be obtained from the corresponding author upon reasonable request.

\printbibliography

@ARTICLE{Freeman1987-rm,
  title     = "Simulation of chaotic {EEG} patterns with a dynamic model of the
               olfactory system",
  author    = "Freeman, Walter J",
  journal   = "Biol. Cybern.",
  publisher = "Springer Science and Business Media LLC",
  volume    =  56,
  number    = "2-3",
  pages     = "139--150",
  year      =  1987,
  language  = "en",
   doi      = {10.1007/BF00317988}
}

@ARTICLE{Freeman2003-ch,
  title     = "Evidence from human scalp electroencephalograms of global
               chaotic itinerancy",
  author    = "Freeman, Walter J",
  journal   = "Chaos",
  publisher = "AIP Publishing",
  volume    =  13,
  number    =  3,
  pages     = "1067--1077",
  month     =  sep,
  year      =  2003,
  language  = "en",
  doi       = {10.1063/1.1596553}
}

@ARTICLE{Ikeda1989-zg,
  title     = "{Maxwell-Bloch} Turbulence",
  author    = "Ikeda, Kensuke and Otsuka, Kenju and Matsumoto, Kenji",
  journal   = "Progr. Theoret. Phys. Suppl.",
  publisher = "Oxford University Press (OUP)",
  volume    =  99,
  pages     = "295--324",
  year      =  1989,
  doi       = {10.1143/PTPS.99.295}
}

@ARTICLE{Inoue2020-hg,
  title     = "Designing spontaneous behavioral switching via chaotic
               itinerancy",
  author    = "Inoue, Katsuma and Nakajima, Kohei and Kuniyoshi, Yasuo",
  journal   = "Sci. Adv.",
  publisher = "American Association for the Advancement of Science (AAAS)",
  volume    =  6,
  number    =  46,
  pages     = "eabb3989",
  month     =  nov,
  year      =  2020,
  language  = "en",
  doi       = {10.1126/sciadv.abb3989},
 eprint     = {https://www.science.org/doi/pdf/10.1126/sciadv.abb3989},
 url        = {www.science.org/doi/pdf/10.1126/sciadv.abb3989}
}

@ARTICLE{Kaneko1990-fj,
  title     = "Clustering, coding, switching, hierarchical ordering, and
               control in a network of chaotic elements",
  author    = "Kaneko, Kunihiko",
  journal   = "Physica D",
  publisher = "Elsevier BV",
  volume    =  41,
  number    =  2,
  pages     = "137--172",
  month     =  mar,
  year      =  1990,
  language  = "en",
  url       = "https://doi.org/10.1016/0167-2789(90)90119-A",
  doi       = {10.1016/0167-2789(90)90119-A}
}

@ARTICLE{Kaneko1991-fh,
  title     = "Globally coupled circle maps",
  author    = "Kaneko, Kunihiko",
  journal   = "Physica D",
  publisher = "Elsevier BV",
  volume    =  54,
  number    = "1-2",
  pages     = "5--19",
  month     =  dec,
  year      =  1991,
  language  = "en",
  doi       = {10.1016/0167-2789(91)90103-G}
}

@ARTICLE{Kaneko1993-uy,
  title     = "Chaotic traveling waves in a coupled map lattice",
  author    = "Kaneko, Kunihiko",
  journal   = "Physica D",
  publisher = "Elsevier BV",
  volume    =  68,
  number    = "3-4",
  pages     = "299--317",
  month     =  oct,
  year      =  1993,
  language  = "en",
  doi       = {10.1016/0167-2789(93)90126-L}
}

@ARTICLE{Kaneko2003-oy,
  title     = "Chaotic itinerancy",
  author    = "Kaneko, Kunihiko and Tsuda, Ichiro",
  journal   = "Chaos",
  publisher = "AIP Publishing",
  volume    =  13,
  number    =  3,
  pages     = "926--936",
  month     =  sep,
  year      =  2003,
  language  = "en",
  doi       = {10.1063/1.1607783}
}

@ARTICLE{Kaneko2015-id,
  title     = "From globally coupled maps to complex-systems biology",
  author    = "Kaneko, Kunihiko",
  journal   = "Chaos",
  publisher = "AIP Publishing",
  volume    =  25,
  number    =  9,
  pages     = "097608",
  month     =  sep,
  year      =  2015,
  language  = "en",
  doi       = {10.1063/1.4916925}
}

@ARTICLE{Kobayashi2018-yv,
  title     = "Chaotic itinerancy observed in mutually coupled Gaussian maps",
  author    = "Kobayashi, Mio and Yoshinaga, Tetsuya",
  journal   = "Int. J. Bifurcat. Chaos",
  publisher = "World Scientific Pub Co Pte Ltd",
  volume    =  28,
  number    =  04,
  pages     = "1830011",
  month     =  apr,
  year      =  2018,
  language  = "en",
  doi       = {10.1142/S0218127418300112}
}

@ARTICLE{Milnor1985-xp,
  title     = "On the concept of attractor",
  author    = "Milnor, John",
  journal   = "Commun. Math. Phys.",
  publisher = "Springer Science and Business Media LLC",
  volume    =  99,
  number    =  2,
  pages     = "177--195",
  month     =  jun,
  year      =  1985,
  language  = "en",
  doi       = {10.1007/BF01212280}
}

@ARTICLE{Park2017-dl,
  title     = "Chaotic itinerancy within the coupled dynamics between a
               physical body and neural oscillator networks",
  author    = "Park, Jihoon and Mori, Hiroki and Okuyama, Yuji and Asada,
               Minoru",
  journal   = "PLoS One",
  publisher = "Public Library of Science (PLoS)",
  volume    =  12,
  number    =  8,
  pages     = "e0182518",
  month     =  aug,
  year      =  2017,
  copyright = "http://creativecommons.org/licenses/by/4.0/",
  language  = "en",
  doi = {10.1371/journal.pone.0182518}
}

@ARTICLE{Sauer2003-xk,
  title     = "Chaotic itinerancy based on attractors of one-dimensional maps",
  author    = "Sauer, Timothy",
  journal   = "Chaos",
  publisher = "AIP Publishing",
  volume    =  13,
  number    =  3,
  pages     = "947--952",
  month     =  sep,
  year      =  2003,
  language  = "en",
  doi       = {10.1063/1.1582332}
}

@ARTICLE{Tanaka2005-if,
  title     = "Crisis-induced intermittency in two coupled chaotic maps:
               towards understanding chaotic itinerancy",
  author    = "Tanaka, G and Sanju{\'a}n, M A F and Aihara, K",
  journal   = "Phys. Rev. E Stat. Nonlin. Soft Matter Phys.",
  publisher = "American Physical Society (APS)",
  volume    =  71,
  number    = "1 Pt 2",
  pages     = "016219",
  month     =  jan,
  year      =  2005,
  copyright = "http://link.aps.org/licenses/aps-default-license",
  language  = "en",
  doi       = {10.1103/PhysRevE.71.016219}
}

@ARTICLE{Tsuda1991-pf,
  title     = "Chaotic itinerancy as a dynamical basis of hermeneutics in brain
               and mind",
  author    = "Tsuda, Ichiro",
  journal   = "World Futures",
  publisher = "Informa UK Limited",
  volume    =  32,
  number    = "2-3",
  pages     = "167--184",
  month     =  sep,
  year      =  1991,
  language  = "en",
  doi       = {10.1080/02604027.1991.9972257}
}

@ARTICLE{Tsuda2003-dx,
  title     = "Chaotic itinerancy generated by coupling of Milnor attractors",
  author    = "Tsuda, Ichiro and Umemura, Toshiya",
  journal   = "Chaos",
  publisher = "AIP Publishing",
  volume    =  13,
  number    =  3,
  pages     = "937--946",
  month     =  sep,
  year      =  2003,
  language  = "en",
  doi       = {10.1063/1.1599131}
}

@ARTICLE{Tsuda2004-mk,
  title     = "Chaotic itinerancy as a mechanism of irregular changes between
               synchronization and desynchronization in a neural network",
  author    = "Tsuda, Ichiro and Fujii, Hiroshi and Tadokoro, Satoru and
               Yasuoka, Takuo and Yamaguti, Yutaka",
  journal   = "J. Integr. Neurosci.",
  publisher = "World Scientific Pub Co Pte Lt",
  volume    =  03,
  number    =  02,
  pages     = "159--182",
  month     =  jun,
  year      =  2004,
  doi       = {10.1142/S021963520400049X}
}

@ARTICLE{Yanagita1995-jm,
  title     = "{Rayleigh-B{\'e}nard} convection patterns, chaos, spatiotemporal
               chaos and turbulence",
  author    = "Yanagita, Tatsuo and Kaneko, Kunihiko",
  journal   = "Physica D",
  publisher = "Elsevier BV",
  volume    =  82,
  number    =  3,
  pages     = "288--313",
  month     =  apr,
  year      =  1995,
  language  = "en",
  doi       = {10.1016/0167-2789(94)00233-G}
}

@BOOK{Hirsch2012,
  title     = "Differential Equations, Dynamical Systems, and an Introduction
               to Chaos",
  author    = "Hirsch, Morris W and Smale, Stephen and Devaney, Robert L",
  publisher = "Academic Press",
  edition   =  3,
  month     =  mar,
  year      =  2012,
  address   = "San Diego, CA",
  language  = "en",
  doi       = {10.1016/C2009-0-61160-0},
  isbn={9780123820112},
  lccn={2012002951}
}

@article{Arai2009,
author = {Arai, Zin and Kalies, William and Kokubu, Hiroshi and Mischaikow, Konstantin and Oka, Hiroe and Pilarczyk, Pawel},
title = {A Database Schema for the Analysis of Global Dynamics of Multiparameter Systems},
journal = {SIAM Journal on Applied Dynamical Systems},
volume = {8},
number = {3},
pages = {757-789},
year = {2009},
doi = {10.1137/080734935},
URL = {https://doi.org/10.1137/080734935}
}

@article{Pilarczyk2023,
    author = {Pilarczyk, Paweł and Signerska-Rynkowska, Justyna and Graff, Grzegorz},
    title = "{Topological-numerical analysis of a two-dimensional discrete neuron model}",
    journal = {Chaos: An Interdisciplinary Journal of Nonlinear Science},
    volume = {33},
    number = {4},
    pages = {043110},
    year = {2023},
    month = apr,
    issn = {1054-1500},
    doi = {10.1063/5.0129859},
    url = {https://doi.org/10.1063/5.0129859},
    eprint = {https://pubs.aip.org/aip/cha/article-pdf/doi/10.1063/5.0129859/16822771/043110\_1\_5.0129859.pdf},
}

@article{Kapela2021,
title = {CAPD::DynSys: A flexible C++ toolbox for rigorous numerical analysis of dynamical systems},
journal = {Communications in Nonlinear Science and Numerical Simulation},
volume = {101},
pages = {105578},
year = {2021},
issn = {1007-5704},
doi = {10.1016/j.cnsns.2020.105578},
url = {https://www.sciencedirect.com/science/article/pii/S1007570420304081},
author = {Tomasz Kapela and Marian Mrozek and Daniel Wilczak and Piotr Zgliczyński},
}

@article{Koscielniak2014,
title = {Shadowing is generic---a continuous map case},
journal = {Discrete and Continuous Dynamical Systems},
volume = {34},
number = {9},
pages = {3591-3609},
year = {2014},
issn = {1078-0947},
doi = {10.3934/dcds.2014.34.3591},
url = {https://www.aimsciences.org/article/id/d7cee083-559f-46cf-aefa-eee5ec88cc47},
author = {Piotr Kościelniak and Marcin Mazur and Piotr Oprocha and Paweł Pilarczyk},
}

@ARTICLE{Luzzatto2011-us,
  title     = "Finite resolution dynamics",
  author    = "Luzzatto, Stefano and Pilarczyk, Pawe{\l}",
  journal   = "Found. Comput. Math.",
  publisher = "Springer Science and Business Media LLC",
  volume    =  11,
  number    =  2,
  pages     = "211--239",
  month     =  apr,
  year      =  2011,
  language  = "en",
  doi       = {10.1007/s10208-010-9083-z}
}

@article{miyaji-2016,
title = {A study of rigorous ODE integrators for multi-scale set-oriented computations},
journal = {Applied Numerical Mathematics},
volume = {107},
pages = {34-47},
year = {2016},
issn = {0168-9274},
doi = {https://doi.org/10.1016/j.apnum.2016.04.005},
url = {https://www.sciencedirect.com/science/article/pii/S0168927416300435},
author = {Tomoyuki Miyaji and Paweł Pilarczyk and Marcio Gameiro and Hiroshi Kokubu and Konstantin Mischaikow}
}

@article{Pilarczyk-2012,
doi = {10.1088/1751-8113/45/12/125502},
url = {https://doi.org/10.1088/1751-8113/45/12/125502},
year = {2012},
month = mar,
publisher = {IOP Publishing},
volume = {45},
number = {12},
pages = {125502},
author = {Pilarczyk, Paweł and García, Luis and Carreras, Benjamin A and Llerena, Irene},
title = {A dynamical model for plasma confinement transitions},
journal = {Journal of Physics A: Mathematical and Theoretical}
}

@techreport{Page1999,
          number = {1999-66},
           month = nov,
          author = {Lawrence Page and Sergey Brin and Rajeev Motwani and Terry Winograd},
            note = {Previous number = SIDL-WP-1999-0120},
           title = {The PageRank Citation Ranking: Bringing Order to the Web.},
            type = {Technical Report},
       publisher = {Stanford InfoLab},
            year = {1999},
     institution = {Stanford InfoLab},
             url = {http://ilpubs.stanford.edu:8090/422/}
}

@book{Cover2006,
  author = {Cover, Thomas M. and Thomas, Joy A.},
  howpublished = {Hardcover},
  id = {1877660},
  isbn = {0471241954},
  month = jul,
  publisher = {Wiley-Interscience},
  title = {Elements of Information Theory 2nd Edition (Wiley Series in Telecommunications and Signal Processing)},
  year = "2006",

}

@article{Mierski2025,
  title = {Analysis of the chaotic itinerancy phenomenon using entropy and clustering},
  author = {Mierski, Nikodem and Pilarczyk, Pawe\l{}},
  journal = {Phys. Rev. E},
  volume = {112},
  issue = {5},
  pages = {054223},
  numpages = {16},
  year = {2025},
  month = nov,
  publisher = {American Physical Society},
  doi = {10.1103/bcwd-nx7r},
  url = {https://link.aps.org/doi/10.1103/bcwd-nx7r}
}

@article{Li_Yorke,
author = {Tien-Yien Li and James A. Yorke},
title = {Period Three Implies Chaos},
journal = {The American Mathematical Monthly},
volume = {82},
number = {10},
pages = {985--992},
year = {1975},
publisher = {Taylor \& Francis},
doi = {10.1080/00029890.1975.11994008},
}

@inproceedings{Jerrum1988,
author = {Jerrum, Mark and Sinclair, Alistair},
title = {Conductance and the rapid mixing property for Markov chains: the approximation of permanent resolved},
year = {1988},
isbn = {0897912640},
publisher = {Association for Computing Machinery},
address = {New York, NY, USA},
url = {https://doi.org/10.1145/62212.62234},
doi = {10.1145/62212.62234},
booktitle = {Proceedings of the Twentieth Annual ACM Symposium on Theory of Computing},
pages = {235--244},
numpages = {10},
location = {Chicago, Illinois, USA},
series = {STOC '88}
}

@article{Kannan2004,
author = {Kannan, Ravi and Vempala, Santosh and Vetta, Adrian},
title = {On clusterings: Good, bad and spectral},
year = {2004},
issue_date = {May 2004},
publisher = {Association for Computing Machinery},
address = {New York, NY, USA},
volume = {51},
number = {3},
issn = {0004-5411},
url = {https://doi.org/10.1145/990308.990313},
doi = {10.1145/990308.990313},
journal = {J. ACM},
month = may,
pages = {497--515},
numpages = {19}
}

@book{EdeHar2010,
  author = {Edelsbrunner, Herbert and Harer, John},
  isbn = {978-0-8218-4925-5},
  pages = {I-XII, 1-241},
  publisher = {American Mathematical Society},
  title = {Computational Topology - an Introduction},
  year = 2010
}

@book{Haggstrom_2002, 
place={Cambridge}, 
series={London Mathematical Society Student Texts}, 
title={Finite Markov Chains and Algorithmic Applications}, 
publisher={Cambridge University Press}, 
author={Häggström, Olle}, 
year={2002}, 
doi={10.1017/CBO9780511613586},
collection={London Mathematical Society Student Texts}}

@ARTICLE{2020SciPy-NMeth,
  author  = {Virtanen, Pauli and Gommers, Ralf and Oliphant, Travis E. and
            Haberland, Matt and Reddy, Tyler and Cournapeau, David and
            Burovski, Evgeni and Peterson, Pearu and Weckesser, Warren and
            Bright, Jonathan and {van der Walt}, St{\'e}fan J. and
            Brett, Matthew and Wilson, Joshua and Millman, K. Jarrod and
            Mayorov, Nikolay and Nelson, Andrew R. J. and Jones, Eric and
            Kern, Robert and Larson, Eric and Carey, C J and
            Polat, {\.I}lhan and Feng, Yu and Moore, Eric W. and
            {VanderPlas}, Jake and Laxalde, Denis and Perktold, Josef and
            Cimrman, Robert and Henriksen, Ian and Quintero, E. A. and
            Harris, Charles R. and Archibald, Anne M. and
            Ribeiro, Ant{\^o}nio H. and Pedregosa, Fabian and
            {van Mulbregt}, Paul and {SciPy 1.0 Contributors}},
  title   = {{{SciPy} 1.0: Fundamental Algorithms for Scientific
            Computing in Python}},
  journal = {Nature Methods},
  year    = {2020},
  volume  = {17},
  pages   = {261--272},
  adsurl  = {https://rdcu.be/b08Wh},
  doi     = {10.1038/s41592-019-0686-2},
}

@article{Seabrook2023,
    author = {Seabrook, Eddie and Wiskott, Laurenz},
    title = {A Tutorial on the Spectral Theory of Markov Chains},
    journal = {Neural Computation},
    volume = {35},
    number = {11},
    pages = {1713-1796},
    year = {2023},
    month = {10},
    issn = {0899-7667},
    doi = {10.1162/neco_a_01611},
    url = {https://doi.org/10.1162/neco_a_01611},
    eprint = {https://direct.mit.edu/neco/article-pdf/35/11/1713/2335627/neco_a_01611.pdf},
}

@article{Falecki2026,
  title = {Topological–numerical analysis of global dynamics in the discrete-time two-gene Andrecut–Kauffman model},
  volume = {36},
  ISSN = {1089-7682},
  url = {http://dx.doi.org/10.1063/5.0333093},
  DOI = {10.1063/5.0333093},
  number = {7},
  journal = {Chaos: An Interdisciplinary Journal of Nonlinear Science},
  publisher = {AIP Publishing},
  author = {Falęcki,  Dorian and Rosman,  Mikołaj and Palczewski,  Michał and Pilarczyk,  Paweł and Bartłomiejczyk,  Agnieszka},
  year = {2026},
  pages = {073150},
  month = Jul
}

@article{Leonov1996,
  title = {On the total variation for functions of several variables and a multidimensional analog of Helly’s selection principle},
  volume = {63},
  ISSN = {1573-8876},
  DOI = {10.1007/bf02316144},
  number = {1},
  journal = {Mathematical Notes},
  publisher = {Springer Science and Business Media LLC},
  author = {Leonov,  A. S.},
  year = {1996},
  month = Jan,
  pages = {61--71}
}

@article{Tsuda2013,
  title = {Chaotic itinerancy},
  volume = {8},
  ISSN = {1941-6016},
  url = {http://dx.doi.org/10.4249/scholarpedia.4459},
  DOI = {10.4249/scholarpedia.4459},
  number = {1},
  journal = {Scholarpedia},
  publisher = {Scholarpedia},
  author = {Tsuda,  Ichiro},
  year = {2013},
  pages = {4459}
}

\end{document}